\documentclass{article}
\usepackage{ijcai16}

\usepackage{times}
\usepackage{helvet}
\usepackage{courier}

\usepackage{amsmath}
\usepackage{amssymb}
\usepackage{amsthm}
\usepackage{algorithm}
\usepackage{algpseudocode}
\usepackage{graphicx}
\usepackage{verbatim}
\usepackage{etoolbox}
\usepackage{array}
\usepackage{subfigure}
\usepackage{multirow}
\usepackage{flushend}

\makeatletter

\makeatletter
\newlength{\trianglerightwidth}
\algnewcommand{\LineComment}[1]{\Statex \hskip\ALG@thistlm $\triangleright$ #1}
\algnewcommand{\LineCommentCont}[1]{\Statex \hskip\ALG@thistlm%
  \parbox[t]{\dimexpr\linewidth-\ALG@thistlm}{\hangindent=\trianglerightwidth \hangafter=1 \strut$\triangleright$ #1\strut}}
\makeatother

\makeatletter
\newenvironment{breakablealgorithm}
  {% \begin{breakablealgorithm}
   \begin{center}
     \refstepcounter{algorithm}% New algorithm
     \hrule height.8pt depth0pt \kern2pt% \@fs@pre for \@fs@ruled
     \renewcommand{\caption}[2][\relax]{% Make a new \caption
       {\raggedright\textbf{\ALG@name~\thealgorithm} ##2\par}%
       \ifx\relax##1\relax % #1 is \relax
         \addcontentsline{loa}{algorithm}{\protect\numberline{\thealgorithm}##2}%
       \else % #1 is not \relax
         \addcontentsline{loa}{algorithm}{\protect\numberline{\thealgorithm}##1}%
       \fi
       \kern2pt\hrule\kern2pt
     }
  }{% \end{breakablealgorithm}
     \kern2pt\hrule\relax% \@fs@post for \@fs@ruled
   \end{center}
  }
\makeatother

\newbool{hidefullversion}
\setbool{hidefullversion}{false}

\newenvironment{fullversion}{}{}
\newenvironment{shortversion}{}{}

\ifbool{hidefullversion}
	{\AtBeginEnvironment{fullversion}{\comment}\AtEndEnvironment{fullversion}{\endcomment}}
	{\AtBeginEnvironment{shortversion}{\comment}\AtEndEnvironment{shortversion}{\endcomment}}

\newtheorem{theorem}{Theorem}

\newtheorem{lemma}{Lemma}

\begin{document}

% In the original styles from ACM, you would have needed to
% add meta-info here. This is not necessary for AAMAS 2015  as
% the complete copyright information is generated by the cls-files.

\title{Catcher-Evader Games\thanks{
	The full version of this paper is available at http://arxiv.org/abs/1602.01896.
	Dmytro contributed to this paper while he was a Ph.D. student at Duke University.
}}

% AUTHORS

% For initial submission, do not give author names, but the
% tracking number, instead, as the review process is blind.


%\numberofauthors{8} %  in this sample file, there are a *total*

%\author{Paper 739}
\author{
	Yuqian Li, Vincent Conitzer, Dmytro Korzhyk \\
	Department of Computer Science, Duke University \\
%s	$^2$ Two Sigma. \\
	\{yuqian, conitzer\}@cs.duke.edu, dima.korzhyk@gmail.com
}

%% There's nothing stopping you putting the seventh, eighth, etc.
%% author on the opening page (as the 'third row') but we ask,
%% for aesthetic reasons that you place these 'additional authors'
%% in the \additional authors block, viz.
%\additionalauthors{Additional authors: John Smith (The Th{\o}rv{\"a}ld Group,
%email: {\texttt{jsmith@affiliation.org}}) and Julius P.~Kumquat
%(The Kumquat Consortium, email: {\texttt{jpkumquat@consortium.net}}).}
%% Just remember to make sure that the TOTAL number of authors
%% is the number that will appear on the first page PLUS the
%% number that will appear in the \additionalauthors section.

\maketitle

\begin{abstract}
Algorithms for computing game-theoretic solutions have recently been
applied to a number of security domains.  However, many of the
techniques developed for compact representations of security games do
not extend to {\em Bayesian} security games, which allow us to model
uncertainty about the attacker's type.  In this paper, we introduce a
general framework of {\em catcher-evader} games that can capture
Bayesian security games as well as other game families of interest.
We show that computing Stackelberg strategies is NP-hard, but give an
algorithm for computing a Nash equilibrium that performs well in
experiments.  We also prove that the Nash equilibria of these games
satisfy the {\em interchangeability} property, so that equilibrium
selection is not an issue.
\end{abstract}

\section{Introduction}

Algorithms for computing game-theoretic solutions have long been of
interest to AI researchers.  In recent years, applications of these
techniques to security have drawn particular attention.  These
applications include airport security~\cite{Pita08:Using}, the
assignment of Federal Air Marshals to flights~\cite{Tsai09:IRIS},
scheduling Coast Guard patrols~\cite{An12:PROTECT}, scheduling
patrols on transit systems~\cite{Yin12:TRUSTSAIMAG},
% protecting wildlife from poachers~\cite{},
and the list goes on.  Game-theoretic techniques are natural in these
domains because they involve parties with competing interests (though
the games are usually not zero-sum), and the use of mixed (randomized)
strategies to avoid being predictable to one's opponent is desirable.

These applications have typically used a {\em Stackelberg} model where
one player (the defender) commits to a mixed strategy first and the
other (the attacker) then optimally responds to this mixed strategy.
Formally, the defender (player $1$) chooses a mixed strategy
$\sigma_1^* \in \arg \max_{\sigma_1} \max_{s_2 \in
\text{BR}^2(\sigma_1)} u_1(\sigma_1,s_2)$,\footnote{Generally, if the attacker
is indifferent among multiple targets, the defender can slightly modify her
strategy to make any one of these uniquely optimal; this is why ties for the
attacker are broken in favor of the defender.}
where
$\text{BR}^2(\sigma_1)$ is the set of best responses to $\sigma_1$ for
player $2$ (i.e., the responses that maximize player $2$'s utility).
This is in contrast to the more standard solution concept of {\em Nash
equilibrium}, where both players play a mixed strategy in such a way
that each plays a best response to the other---that is, a pair
$(\sigma_1,\sigma_2)$ with $\sigma_1 \in \text{BR}^1(\sigma_2)$ and
$\sigma_2 \in \text{BR}^2(\sigma_1)$. 
Arguably, the Stackelberg solution is well motivated in contexts where the
attacker can learn the defender's strategy over time by repeated observation,
whereas if this is not the case perhaps the Nash solution is better motivated.
It is known that under certain
conditions in security games, Stackelberg strategies are also Nash equilibrium
strategies~\cite{Korzhyk11:Stackelberg}.

Initial work in these domains modeled uncertainty over attacker
preferences using the formalism of Bayesian games, assigning
probabilities to different types of
attackers. This included the original work at the airport at Los Angeles~\cite{Paruchuri08:Playing}.
However, subsequent research,
which started to focus on compact representations of security games,
mostly did not consider Bayesian games.  In this paper, we introduce a
more general framework that can capture such Bayesian security games,
and study the computation of Stackelberg and Nash solutions in them
(which in such games generally do not coincide).  Our framework can
also model certain types of {\em test games} in which a tester
randomly chooses questions from a fixed database of
questions~\cite{Li13:Game}.  We show that computing a Stackelberg
strategy is strongly NP-hard, but give an algorithm for computing Nash
equilibria that combines and expands on earlier techniques in both
security and test games.  While we have been unable to show that
our algorithm is guaranteed to require at most polynomially many
iterations, it requires few iterations in experiments.

More benefits of our framework are listed below:
%\begin{enumerate}
(1) Our notation for Catcher-Evader\footnote{Note that these games are
completely different from {\em pursuit-evasion} (or {\em cops-and-robbers})
games~\cite{Parsons78:Pursuit,Borie09:Algorithms}.  Those games involve dynamically chasing another player on a
graph.  Our games, in contrast, occur in a single period, and concern the
computation of an optimal random assignment.} games, once one becomes familiar with
it,
	greatly simplifies analysis of those games, especially as it concerns utilities.
	For example, our notation expresses the utility delta of a target, which is often
	the crucial quantity, directly as $d$, rather than as a difference (e.g., $u_i^c - u_i^u$).
(2) Our additional parameters $a, b, c$ allow richer utility functions that
	security games did not capture previously. For example, targets may have
	different costs to defend even if the attacker does not attack them.
	Previous security game definitions always assumed no cost (or the same cost) if the attacker
	does not attack.
	%Moreover, the equilibria of games in which targets have different costs to
	%defend are structurally different, hence there exists no simple transformation
	%to the previously studied (costless) case.
(3) It lets us swap the roles of defenders and attackers.
	Therefore, we can also directly compute the attacker's strategy as well as
	the defender's strategy, an example of which is computing the tester's strategy
	in test games.
(4) Its connection between security games and test games brings enormous convenience
	for algorithm design. Previously, separate algorithms had to be designed for them,
	but now we can design a single algorithm for both. Moreover,
	we can potentially apply known algorithms for each of these game families to the other.
	For example, the aforementioned Nash equilibria algorithm combines techniques for security games
	(progressively increasing defender or catcher resources) and test games
	(using network flow to reallocate attacker or evader resources).
(5) Besides security games and test games, it can also capture other interesting scenarios
	where resources must be assigned to different targets by two competing parties.
	For example, two companies, an incumbent and an entrant, might be allocating capital to
	different markets; the entrant may wish to evade the incumbent and build up market share,
	while the incumbent wants to catch the entrant to drive the latter out of business.

%\end{enumerate}

\section{Notation}

We model a Catcher-Evader game (CE game) as a game between one
catcher and multiple evaders.  Since we assume that the evaders do not
care about each other's actions, this is equivalent to a Bayesian game between a
single-typed catcher and an evader with multiple types. Also, as we will show
in section~\ref{subsection:swapping}, the roles of catcher and evader can be
swapped. Hence, our model also captures games between one evader and multiple
catchers.

%We model the Catcher-Evader game (CE game for short) as a 2-player game
%between a catcher and an evader. The catcher has a fixed type and usually leads
%in Stackelberg settings (e.g., the defender in security games).  The evader,
%however, may be Bayesian with multiple types where the type $i$ evader occurs
%with probability $p_i$.  This is equivalent to a game between one catcher and
%multiple evaders indexed by $i$ where evader $i$ has weight $p_i$ to the
%catcher. Also, as we later show in Section~\ref{section:reduction}, the roles
%of catcher and evader can be swapped. So our model also captures games
%between one evader and multiple (types of) catchers (e.g., test games).

We represent a CE game by $(N, \Psi, r, \ell, a, b, c, d)$, where $N = \{0,
1, \ldots, n\}$ is the set of players and $\Psi$ is the set of {\em sites}
(e.g., the targets in a security game or the questions in a test game).
We fix $0 \in N$ to be the catcher (e.g., the defender in a security game),
and $N^+=\{1, 2, \ldots, n\}$ to be the set of evaders (e.g., the multiple types of
attackers in a security game).  Player $i \in N$ has available a total
resource amount of $r_i \in \mathbb R^{\geq 0}$. 
For example, we might set $r_i=1$ to indicate that $i$ has only
one resource, or we might set $r_i=1/2$ to indicate that, in a Bayesian game, a type $i$ that
appears with probability $1/2$ has only a single resource, and therefore
the expected number of resources that this type contributes is $1/2$.
This resource amount can be split fractionally across the sites, for
example, $1/3$ could be assigned to one site and $2/3$ to another.
(This would typically correspond to assigning a single resource to the
former site with {\em probability} $1/3$.)
% (corresponding to the vector $r$ in the definition).
Player $i$ can assign a resource amount of at most $\ell_{i, \psi} \in
\mathbb R$ to site $\psi \in \Psi$.
% (corresponding to the matrix $\ell$
For example, we might set $\ell_{i, \psi}=1$ to indicate that $i$ can
assign at most a single resource to $\psi$, or we might set $\ell_{i,
  \psi}=1/2$ to indicate that, in a Bayesian game, a type $i$ that appears
with probability $1/2$ can assign at most a single resource to $\psi$ if he
appears, and therefore his marginal contribution of probability mass to
$\psi$ is at most $1/2$.  Generally, $r_i \leq \sum_{\psi \in \Psi}
\ell_{i, \psi}$ so the player has to make a nontrivial decision about which
site gets more of the resource amount and which one gets less.

Finally, the utility is encoded by $a, b, c, d$ as follows.  Let $x$ be the
strategy profile where $x_{i,\psi}$ is the resource amount that player $i$
puts on site $\psi$. For convenience, we denote $x_{\Sigma,\psi} =
\sum_{i=1}^n x_{i,\psi}$ as the combined resource amount that all $n$
evaders put on site $\psi$. Then the utility is $\sum_{\psi \in \Psi}
\left[ (b_{0,\psi} + d_{0,\psi} x_{\Sigma,\psi})x_{0,\psi} + a_{0,\psi}
  x_{\Sigma,\psi} + c_{0,\psi} \right]$ for the catcher and $\sum_{\psi
  \in \Psi} \left[ (b_{i,\psi} + d_{i,\psi} x_{0,\psi})x_{i,\psi} +
  a_{i,\psi} x_{0,\psi} + c_{i,\psi} \right]$ for evader $i$.
%We call $b$ the {\em base utility} and $d$ the {\em utility change (delta)}.
Here, $b$ is the {\em base utility} for a player to put a resource at a site,
and $d$ is the {\em utility change} that results from putting a resource at
that site when the opponent puts a resource there as well.
Since $c$
({\em constant utility}) is not affected by any player's strategy, we can
ignore it (or let $c=0$) without affecting our analysis of both Stackelberg
strategies and Nash equilibrium.
%Also, players cannot directly affect
%their utilities associated with $a$ ({\em alternating utility}) using their
%own strategies.
Finally, $a$ (for {\em alternating utility}) is the utility that a player
receives when the opponent puts a resource at that site; the former player
cannot affect this.
Hence, for Nash equilibrium (but not for Stackelberg
strategies), we can simply drop $a$ (or let $a = 0$). We require
$\sum_{\psi \in \Psi} x_{i,\psi} = r_i$ for feasibility, as well as $d_{0,\psi} >
0$ and $d_{i,\psi} < 0$ for $i \in N^+$ so that the catcher wants to catch the
evader while the evader wants to evade.

For convenience, we define $x_{-0,\psi} = x_{\Sigma,\psi}$ and $x_{-i,\psi}
= x_{0,\psi}$ for $i \in N^+$. %, as commonly did in game theory
                             %literatures. 
Then, we define $\mu_{i, \psi} = (b_{i,\psi} + d_{i,\psi} \cdot
x_{-i,\psi})$ as the {\em per-resource utility} of player $i$ on site
$\psi$.
That is, it is the increase in utility she experiences from putting one more resource there.
So, player $i$'s utility gained from site $\psi$ can be written as
$u_{i,\psi}(x) = \mu_{i,\psi} x_{i,\psi} + a_{i,\psi} x_{-i,\psi} +
c_{i,\psi}$. In a best-response strategy, player $i$ should have a utility
threshold $\theta_i$ such that (1) for all $\psi$ with $\mu_{i,\psi}(x) >
\theta_i$, the player maximizes the resource amount it puts there ($x_{i,\psi} =
\ell_{i,\psi}$), and (2) for all $\psi$ with $\mu_{i,\psi}(x) < \theta_i$, the
player puts no resource amount there ($x_{i,\psi}=0$).
(There is no requirement for the case $\mu_{i,\psi}(x)=\theta_i$.)
The value of $\theta_i$ is not
necessarily unique, so for definiteness, let
$\theta_0 = \max_{\psi \in \Psi: x_{0,\psi} < \ell_{0,\psi}}
\mu_{0,\psi}$ and $\theta_i = \min_{\psi \in \Psi: x_{i,\psi} > 0}
\mu_{i,\psi}$ for $i \in N^+$.
%\footnote{Sometimes, multiple values could
%  serve as a single player's threshold; we choose the given one for
%  analysis conveniences.}

Incidentally, note that if we do not require $d_{0,\psi} > 0$ and $d_{i,\psi} <
0$ for $i \in N^+$, then $a, b, c, d$ can represent any utility function of the
form  $\sum_{\psi \in \Psi} f(x_{i,\psi}, x_{-i, \psi})$ where $f$ is a 
quadratic polynomial without factors $x_{i,\psi}^2$ or $x_{-i, \psi}^2$. 
%or
%higher-degree terms.

In Table~\ref{table:symbols}, we summarize all symbols for reference.

\begin{table}
\small
\center
\begin{tabular}{c | p{0.78\linewidth}}
 & Description \\
\hline
$N$ & Set of players $\{0, 1, \ldots, n\}$ \\
$N^+$ & Evaders $\{1, 2, \ldots, n\}$ ($0$ is the catcher) \\
$\Psi$ & Set of {\em sites} (e.g., targets in security games) \\
$r_i$ & Resource of player $i$ \\
$\ell_{i, \psi}$ & Resource {\em limit} player $i$ can put on site $\psi$ \\
$a_{i, \psi}$ & Alternating utility of player $i$ on site $\psi$ \\
$b_{i, \psi}$ & Base utility of player $i$ on site $\psi$ \\
$c_{i, \psi}$ & Constant utility of player $i$ on site $\psi$ \\
$d_{i, \psi}$ & Utility change ({\em delta}) of player $i$ on site $\psi$ \\
$x_{i, \psi}$ & Amount of resource $i$ puts on $\psi$ ({\em strategy}) \\
$x_{\Sigma, \psi}$ & Sum of all evaders' resource on $\psi$ \\
$x_{-i, \psi}$ & Amount of resource $i$'s opponent puts on $\psi$ \\
$\mu_{i, \psi}$ & Per-resource utility of $i$ on $\psi$:
					$b_{i,\psi} + d_{i,\psi} x_{-i,\psi}$ \\
$u_{i, \psi}$ & Utility of $i$ on $\psi$:
			$\mu_{i,\psi} x_{i,\psi} + a_{i,\psi} x_{-i, \psi} + c_{i,\psi}$ \\
$\theta_i$ & Utility threshold of player $i$
\end{tabular}
\caption{Symbols used for CE games.}
\label{table:symbols}
\end{table}

\section{Reducing Games to CE Games}
\label{section:reduction}

In this section, we show how the framework of CE games let us capture
several game families studied previously in the literature, namely security
games and test games.

\subsection{Security Games}
\label{subsection:reduce_security}

A general definition of security games was given
by~\cite{Kiekintveld09:Computing}.  That work considered only a single
attacker resource; an attacker with multiple attacker resources was
considered by~\cite{Korzhyk11:Security}.  More generally still, we can
consider a Bayesian game in which there is uncertainty about the type of
the attacker.
% (equivalently, where there are multiple attackers whose
%utilities do not depend on the actions of the other attackers).  
(Some of the earliest work in this line of research concerned Bayesian
games~\cite{Paruchuri08:Playing,Pita09:Using}, but the games were
relatively small and so the techniques did not exploit the structure of
security games.)  We now define multi-resource Bayesian security games and
show how to reduce them to CE games.  Note that in our definition, a
resource is assigned to a single target.\footnote{
Section 6 of~\cite{Kiekintveld09:Computing} also allowed resources to be
assigned to {\em schedules} of multiple targets, which quickly leads to
NP-hardness~\cite{Korzhyk10:Complexity}.}

There are a defender and an attacker.  The latter has unknown type $i \in
\{1, \ldots, n\}$.  An attacker of type $i$ occurs with probability
$p_i$. There are $m$ targets $t_1, t_2, \ldots, t_m$.  An attacker of type
$i$ can attack $r_i$ distinct targets while the defender can defend $r_d$
distinct targets. A player's utility is the sum of its utility over
all targets.  If an attacker of type $i$ attacks an undefended target $t$,
it obtains utility $u_i^u(t)$ (and the defender obtains utility  $u_d^u(t)$).
  If it attacks a defended (covered)
target $t$, it obtains utility $u_i^c(t)$ (and the defender obtains utility
$u_d^c(t)$). 
%For the defender, if target $t$
%is attacked while undefended, she receives $u_d^u(t)$ utility. If $t$ is
%defended and attacked, she receives $u_d^d(t)$. 
Both players obtain utility $0$ from $t$ if $t$ is unattacked.

Now, we can reduce this to the following CE game $(N, \Psi, r', a', b',
c'=0, d')$ (see Table~\ref{table:security} for an example of utility reduction):
$
	N = \{0, 1, 2, \ldots, n\}, \Psi = \{t_1, t_2, \ldots, t_m\},
	r'_0 = r_d, r'_i = p_i r_i ~(i \in N^+),
	\ell'_{0,\psi} = 1, \ell'_{i, \psi} = p_i ~(\psi \in \Psi, i \in N^+),
	a'_{0, \psi} = u_d^u(\psi), b'_{0,\psi} = 0, d'_{0,\psi} 
		= u_d^c(\psi)-u_d^u(\psi),
	a'_{i, \psi} = 0, b'_{i,\psi} = u_i^u(\psi), d'_{i,\psi} 
		= u_i^c(\psi)-u_i^u(\psi) ~(i \in N^+)
$.

Note that in the original security game, $r$ consists of natural numbers
and a pure strategy would put either $0$ or $1$ resources on each site.  In
the CE game, the strategy profile $x_{i,\psi}$ corresponds to the marginal
probability that player $i$ puts a resource on $\psi$.  Because resources
can only be assigned to single targets, we can always use Birkhoff-von
Neumann decomposition~\cite{Birkhoff46:Tres} to generate a valid mixed
strategy of the original security game with these marginals (see
also~\cite{Korzhyk10:Complexity}).

\begin{table}
\centering
\tiny
\begin{tabular}{r | c c | c c c c}
	Player & \multicolumn{2}{c|}{
			Security Game
			%Security Game's Spec. on Target $t$
		}
		& \multicolumn{4}{c}{
			CE Game
			%Reduced CE Game's Spec. on Site $\psi = t$
		}\\
	 & $u_i^c(t)$ & $u_i^u(t)$ & $a_{i,t}$ & $b_{i,t}$ & $c_{i,t}$ & $d_{i,t}$ \\
	 \hline
	 Def ($i=0$) & 1 & -10 & -10 & 0 & 0 & 11 \\
	 Att 1 ($i=1$) & -5 & 5 & 0 & 5 & 0 & -10 \\
	 Att 2 ($i=2$) & -9 & 10 & 0 & 10 & 0 & -19
\end{tabular}
\caption{Example of how a security game's utility specification for a target $t$
is converted to a CE game's utility specification for a site $\psi = t$. In this
table, we let $u_0^c(t) = u_d^c(t), u_0^u(t) = u_d^u(t)$ for convenience.}
\label{table:security}
\end{table}

\subsection{Testing Games}

Testing games were recently studied by~\cite{Li13:Game}.  In that work,
only test takers that do not fail any questions pass the test; therefore,
it does not matter whether a test taker fails $1$ question or $100$.  In
contrast, we consider a variant---arguably more realistic---in which the
losses and gains the players experience are additive across questions.
We call this variant ``scored tests'', which captures cases like the GRE, the TOEFL,
and most course exams at school.
It allows us to bypass the (co)NP-hardness results for
computing the best test strategies from~\cite{Li13:Game}.  
On the other hand, the transformation to a zero-sum game described in that paper
no longer works in this context.

%The test games we consider here are based on a previous study~\cite{}, with
%the following change: when a test taker failed multiple questions, he suffers
%multiple losses. In the previous study, such loss is equal to that of failing a
%single question because of the binary result (pass or fail the whole test).
%This change not only better describes most real-world tests where each
%question has a score, but also helps us bypass the (co)NP hardness of
%computing the best test strategies proved in the previous study.
%Nevertheless, this change also poses new challenges to the computation: with
%non-binary result, the zero-sum transformation that was exploited in the
%previous study no longer works. (YUQIAN: MORE EXAMPLES AND EXPANATIONS?)

Formally, a test game is a 2-player game between a tester and a test
taker.  The tester is uncertain about the test taker's type $i \in \{1, 2,
\ldots, n\}$, but she knows that a test taker of type $i$ occurs with
probability $p_i$.  The tester has a pool of questions $Q$, from which $t$
questions will be chosen to form a test $T \subseteq Q$ ($|T| = t$). For a
test taker of type $i$, a given subset $H_i \subseteq Q$ of questions are
hard and he will not be able to solve them unless he memorizes their
answers (or writes them on a cheat sheet). However, he can memorize at most
$m_\theta$ questions, so if the tester randomizes over the choice of $T$,
there is a good chance that most questions in $T$ have not been memorized.
We denote the set of questions $i$ chooses to memorized as $M_i \subseteq Q
(|M_i| = m_i)$

So far, everything is identical to the games defined by~\cite{Li13:Game}.
Now we introduce a question score $s_q$ for each $q \in Q$. If a test taker fails
to solve $q$ in the test, $s_q$ is deducted from his score. Hence the test taker's
utility is $u_i(T, M_i) = -\sum_{q \in T \cap H_i \setminus
  M_i} s_q$.\footnote{A constant $\sum_{q \in T} s_q$ can be added to
  $u_i(T, M_i)$ to obtain the usual nonnegative test scores.}
We also introduce a weight $w_q$ for each question, representing how
important the tester thinks it is to find out whether the test taker can
solve $q$. This may or may not be equal to $s_q$.  
%For
%example, there might be a very easy $q$ with a low score $s_q$, but it is
%extremely important (high $w_q$) for the tester to know whether the test
%taker even fails to solve this easy question.  
The tester's utility is then $u^t_i(T, M_i) = v_i \sum_{q \in T \cap H_i
  \setminus M_i} w_q$. Here, $v_i$ denotes the tester's assessment of the
importance of test taker type $i$. For example, it might be more (or less)
important to figure out the true score of a bad test taker (with large
$H_i$) than that of a good one.
%Without loss of generality, we assume $\Theta = \{\theta_1, \theta_2, \ldots,
%\theta_n\}$. 
 We reduce this game to the CE game $(N, \Psi, r, a, b, c=0, d)$
where
$
N = \{0, 1, 2, \ldots, n\}, \Psi = Q,
	r_0 = t, r_i = p_{i} v_{i} m_{i} \ (i \in N^+),
	\ell_{0,q} = 1, \ell_{i,q} 
		= p_{i} v_{i} \ (i \in N^+, q \in Q=\Psi),
	a_{0,q} = 0, 
		b_{0,q} = w_q \sum_{i: q \in H_{i}} p_{i} v_{i},
		d_{0,q} = -w_q,
	a_{i,q} = -s_q \text{ for } q \in H_{i}, a_{i,q} = 0 \text{ for } q \notin H_i,
	b_{i,q} = 0 \ (i \in N^+),
	d_{i,q} = s_q / r_{i,q} \text{ for } q \in H_{i},
		d_{i,q} = 0 \text{ for } q \notin H_i
$.

Similar to security games, the resulting strategy profile $x_{i,q}$ denotes
the marginal probability that a player puts $q$ on the test / memorizes
$q$; again, the Birkhoff-von Neumann theorem allows us to obtain a strategy
with these marginals.
%that a player put one resource on $q$: for
%the tester, that means testing $q (q \in T)$; for the test taker, that means
%memorizing $q (q \in M_\theta)$.

\subsection{Swapping Roles}\label{subsection:swapping}

The reduction from test games has one issue: the utilities change at
rates $d_0 < 0, d_i > 0 \ (i \in N^+)$ but CE games require $d_0 > 0, d_i <
0 \ (i \in N^+)$. In a sense, the tester is an evader who wants to evade
by asking questions that are not memorized by the test taker; but as we
have defined them, in CE games, player $0$ is a catcher.

We handle this by redefining player $0$'s resources to their opposites.
That is, we focus on which questions she does {\em not} test.
%For example, let the tester's resources be the number of questions that
%she {\em cannot} test. 
Hence, the modified $x'_{0,q}$ will be the
marginal probability that she does {\em not} test $q$ (i.e., $q \notin T$).
%Note that the test taker's resource is not changed:
%\begin{align*}
%	N &= \{0, 1, 2, \ldots, n\}, \Psi = Q\\
%	r_0 &= |Q|-t, r_i = p_{\theta_i} v_{\theta_i} m_{\theta_i}
%		(1 \leq i \leq n)\\
%	r_{0,q} &= 1, r_{i,q} 
%		= p_{\theta_i} v_{\theta_i} (1 \leq i \leq n, q \in Q=\Psi)\\
%	a_{0,q} &= -w_q, 
%		b_{0,q} = -w_q \sum_{i: q \in H_{\theta_i}} p_{\theta_i} v_{\theta_i},
%		d_{0,q} = w_q\\
%	a_{i,q} &= s_q \text{ for } q \in H_{\theta_i},
%			a_{i,q} = 0 \text{ otherwise }\\
%	b_{i,q} &= s_q r_{0,q} / r_{i,q} \text{ for } q \in H_{\theta_i},
%			b_{i,q} = 0 \text{ otherwise}\\
%	d_{i,q} &= -s_q / r_{i,q} \text{ for } q \in H_{\theta_i},
%			d_{i,q} = 0 \text{ otherwise }
%\end{align*}

In general, we can swap roles between catchers and evaders (i.e., negate $d$) by rewriting CE game 
$(N, \Psi, r, a, b, c, d)$ as CE game $(N, \Psi, r', a', b', c', d')$:
$
	r'_0 = - r_0 + \sum_{\psi \in \Psi} \ell_{0,\psi}, ~r'_i = r_i \ (i
        \in N^+),
	\ell'_{i,\psi} = \ell_{i,\psi} \ (i \in N),
	a'_{0,\psi} = a_{0,\psi} + d_{0,\psi} \ell_{0,\psi},
		c'_{0, \psi} = c_{0,\psi} + b_{0,\psi} \ell_{0,\psi},
	b'_{0, \psi} = -b_{0, \psi}, d'_{0, \psi} = -d_{0, \psi},
	a'_{i,\psi} = -a_{i,\psi}, d'_{i,\psi} = -d_{i,\psi},
	b'_{i,\psi} = b_{i,\psi} + d_{i,\psi} \ell_{0,\psi},
		c'_{i,\psi} = c_{i,\psi} + a_{i,\psi} \ell_{0,\psi}	
$

\begin{fullversion}
The correctness of this transformation is the result of the following
equations, letting $x'_{0,\psi} = \ell_{0,\psi}-x_{0,\psi}$ (note
$x'_{i,\psi} = x_{i,\psi}$ for $i \in N^+$):
\begin{align*}
	&\left[ (b'_{0,\psi} + d'_{0,\psi} x_{\Sigma,\psi})x'_{0,\psi}
			+ a'_{0,\psi} x_{\Sigma,\psi} + c'_{0,\psi} \right]\\
	&=\left[ (b_{0,\psi} + d_{0,\psi} x_{\Sigma,\psi})x_{0,\psi}
			+ a_{0,\psi} x_{\Sigma,\psi} + c_{0,\psi} \right]\\
	&\left[ (b'_{i,\psi} + d'_{i,\psi} x'_{0,\psi})x'_{i,\psi}
			+ a'_{i,\psi}x'_{0,\psi} + c'_{i,\psi} \right]\\
	&=\left[ (b_{i,\psi} + d_{i,\psi} x_{0,\psi})x_{i,\psi}
			+ a_{i,\psi} x_{0,\psi} + c_{i,\psi} \right]\\
\end{align*}
\end{fullversion}

Hence, the utilities are exactly the same as in the original game.
As previously mentioned, $c$ does not affect our game-theoretic analysis.
However, it is essential for establishing these equations so we can swap
roles. Of course, after the transformation, we can freely drop $c'$.
Table~\ref{table:swap} shows an example of a test game and
how we swap roles in it.

\begin{table}
\centering
\tiny
\subtable[An example of test game players' utility on a question $q$]{
\begin{tabular}{c | c | c }
 test taker's utility, tester's utility & don't test $q$ & test $q$ \\
 \hline
 don't memorize $q$ & 0, 0 & -5, 4\\
 \hline
 memorize $q$ & 0, 0 & 0, 0\\
\end{tabular}
}
\subtable[Swapping roles for the above example test game]{
\begin{tabular}{p{0.16\linewidth} | r | c c c c}
	& Player & $a_{i,q}$ & $b_{i,q}$ & $c_{i,q}$ & $d_{i,q}$ \\
	\hline
	\multirow{2}{\linewidth}{
		test $q$: $x_{0,q}=1$
		%Before swapping role: $x_{0,q}=1$ is testing $q$
		%No swap%: test $x_{0,q}=1$
	} & Tester ($i=0$) & 0 & 4 & 0 & -4 \\
	& Test taker ($i=1$) & -5 & 0 & 0 & 5 \\
	\hline
	\multirow{2}{\linewidth}{
		test $q$: $x_{0,q}=0$
		%Before swapping role: $x_{0,q}=1$ is testing $q$
		%Swap%: test $x_{0,q}=1$
	} & Tester ($i=0$) & -4 & -4 & 4 & 4 \\
	& Test taker ($i=1$) & 5 & 5 & -5 & -5 \\
\end{tabular}
}
\caption{Example of a test game and role swapping.}
\label{table:swap}
\end{table}

\section{Complexity of Stackelberg Strategies}

\begin{theorem}\label{theorem:P}
  If there is only one evader who can put all resources on any
  single site ($\forall \psi \in \Psi,~ \ell_{1,\psi} \geq r_1$), then
  catcher Stackelberg strategies can be computed in polynomial time.
\end{theorem}

The proof of Theorem~\ref{theorem:P}
\begin{shortversion}
(in the full version of this paper)
\end{shortversion}
uses a by now fairly standard linear program technique.

\begin{fullversion}

\begin{proof}
  There exists an optimal solution where the evader will assign all its
  resources to the same best-response site $\psi^*$.  For each such site
  $\psi^*$, we can write a linear program that produces the optimal Stackelberg
  strategy under the constraint of $\psi^*$ being a best response; the best of
  these solutions overall will be the Stackelberg strategy.  (See
  also~\cite{Conitzer06:Computing,Korzhyk10:Complexity}.)
\end{proof}

\end{fullversion}

In contrast, it has been shown that computing Stackelberg strategies in a
multi-resource security game (even with only a single type, i.e.,
non-Bayesian) is (weakly) NP-hard~\cite{Korzhyk11:Security}.  Hence, by
our reduction of such security games to CE games, even if the CE game has
only one evader ($n=1$), it is (weakly) NP-hard to compute Stackelberg
strategies if we allow $\ell_{1,\psi} < r_1$ (so the evader/attacker will
put resources on multiple sites).

Next, we show that even if $\ell_{i,\psi} \geq r_i$ for all $i \geq 1$, it
is strongly NP-hard to compute Stackelberg strategies if we allow $n > 1$.  This
corresponds to the case of a Bayesian security game in which each attacker
has only a single resource. Note that the initial LAX airport paper~\cite{Paruchuri08:Playing}
assumed a Bayesian security game with a single attacker resource.
To our best knowledge, no hardness result has
been given for computing Stackelberg strategies of such games.
Also, unlike the known weak
NP-hardness result for multiple resources, this rules out
pseudopolynomial-time algorithms.
\begin{shortversion}
The proof is in the full version of this paper to save space.
\end{shortversion}

\begin{theorem} \label{theorem:Stackelberg}
  Computing Stackelberg strategies in Bayesian security games is strongly
  NP-hard even if each attacker type has only a single resource.
  Consequently, computing Stackelberg strategies in a Catcher-Evader game
  is strongly NP-hard (if $n > 1$), even if $\ell_{i,\psi} \geq r_i$ for all
  $i \in N^+$. (This result is tight in the sense that this problem is also in NP.)
\end{theorem}

\begin{fullversion}
\begin{proof}
%We focus on the $n>1$ and $\forall 1 \leq i \leq n,~ \ell_{i,\psi} \geq
%r_i$ case (the $n=1$ case is proved in the previous work~\cite{}).  
  The reduction is from Satisfiability.  In a Satisfiability instance,
  there are $n$ boolean variables and $m$ clauses. Each clause includes a
  subset of the variables and/or their negations. The 
  problem is to decide whether there is an assignment of true/false values
  to the variables such that each clause has at least one literal set to
  true.
We reduce a Satisfiability instance to a security game as follows.

%, which can be further reduced
%to a CE game such that $\forall 1 \leq i \leq n,~ \ell_{i,\psi} \geq r_i$,
%using the reduction in subsection~\ref{subsection:reduce_security}.

There are $2n+2$ targets: 

(a) $2n$ targets corresponding to the variables and
their negations. We will call these \emph{variable-targets}. The defender gets a utility of $0$ if any of these targets
is attacked, no matter if it is defended or not.\\
(b) One \emph{punishment-target}. The defender gets $-\infty$ utility
if this target is attacked, no matter if it is defended or not.\\
(c) One \emph{bonus-target}. The defender gets a utility of $1$
if this target is attacked, no matter if it is defended or not.

The defender has $n$ resources. The attacker types will be set up in such a
way that if the Satisfiability instance has a solution then the optimal
defender strategy is to defend the targets which correspond to the
negative-valued literals with probability $1$, and the defender gets a
utility of $u^\ast$ as a result. (The value of $u^\ast$ will be defined
below.) If the Satisfiability instance has no solution then the defender's
Stackelberg utility is necessarily below $u^\ast$.

There are $3n+m$ attacker types:

(a) $2n$ attacker types whose job is to count how many of the
  variable-targets are covered with probability $1$ (\emph{counting-types}). Each of these attacker
  types is interested in one variable-target and the bonus-target. If the
  variable-target is covered with probability $1$ then the attacker
  (weakly) prefers
  the bonus-target and the defender gets a utility of $1$; otherwise, the
  attacker prefers the variable-target and the defender gets a utility of
  $0$.\\
(b) $n$ attacker types which make sure that for literals $x_i$ and $\neg
  x_i$, no more than one corresponding target is covered with probability
  $1$ (\emph{paired types}). Each of these attacker types is interested in
  one pair of variable-targets and in the punishment-target. If both
  variable-targets are defended with positive probability then the attacker
  chooses the punishment target, and the defender gets a utility of
  $-\infty$.  Otherwise, the attacker (weakly) prefers the variable-target
  with $0$ probability.\\
(c) $m$ attacker types which check whether the clauses are satisfied
  (\emph{clause-types}). Each clause-type is interested in all targets
  corresponding to the variables in a clause and in the punishment
  target. If all the targets in the clause are defended with positive
  probability (meaning their values are all false) then the attacker
  chooses the punishment target and the defender gets a utility of
  $-\infty$. If there is any target defended with probability $0$ in the
  clause then the attacker (weakly) prefers that target and the defender
  gets $0$ (meaning the clause is satisfied).

The probability of each type is $1/(3n+m)$.  We claim that the defender's
optimal utility is $u^\ast = n/(3n+m)$ if and only if the Satisfiability
instance has a solution.

\textbf{(satisfiable $\Rightarrow$ $u^*$ is feasible)} The defender can
defend, with probability $1$, the targets that correspond to the literals
set to false. All paired-types and clause-types will attack
variable-targets, from which the defender gets $0$ utility. Exactly $n$ of
the counting-types will attack the bonus target (each giving the defender a
utility of $1$) and the remaining $n$ counting-types will attack
variable-targets (each giving the defender a utility of $0$). Hence, the
defender's total utility will be $n/(3n+m)$.

\textbf{($u^*$ is feasible $\Rightarrow$ satisfiable)} Note that only the
bonus-target gives the defender  positive utility, and only the
counting-types are interested in the bonus-target. For the defender to get
$n/(3n+m)$ or more, there must be $n$ or more counting-types attacking
the bonus target. That means there must be $n$ or more variable-targets
defended with probability $1$. Since the defender has only $n$ resources,
there must be exactly $n$ variable-targets defended with probability $1$,
and all the other targets must be defended with probability $0$. The
paired-types enforce that no two of those targets correspond to a variable
and its negation. The clause-types enforce that each clause has at least
one target defended with probability $0$. Hence, to get the Satisfiability
solution, we can set to true the variables corresponding to the targets
defended with probability $0$.
\end{proof}
\end{fullversion}

\section{Interchangeability of NE}
\label{se:interchangeability}

We now move on to studying Nash equilibria.  In general, a downside of the
Nash equilibrium concept is that Nash equilibria can fail {\em
  interchangeability}: if one player plays according to one Nash
equilibrium and the other according to another, the result may not be a
Nash equilibrium.  However, it has been shown that interchangeability of
Nash equilibria is guaranteed in security games and test games under
certain
conditions~\cite{Korzhyk11:Stackelberg,Korzhyk11:Security,Li13:Game}.  We
now show that this also holds for CE games.
\begin{shortversion}
The key lemma and theorem are shown below. Their proofs are in the
full version to save space.
\end{shortversion}

\begin{lemma} \label{lemma:same}
  For each site $\psi$, either $x_{0,\psi}$ is the same
  for all NE or $x_{\Sigma,\psi}$ is the same for all NE.
%That is, there
%  cannot be two NE $x$ and $x'$ and some $\psi$ such that $x_{0,\psi} \neq
%  x_{0,\psi}'$ and $x_{\Sigma,\psi} \neq x_{\Sigma,\psi}'$.
\end{lemma}

\begin{fullversion}
\begin{proof}
We first show that the lemma is equivalent to (A) there
cannot be two NE $x$ and $x'$ and some $\psi$ such that $x_{0,\psi} \neq
x_{0,\psi}'$ and $x_{\Sigma,\psi} \neq x_{\Sigma,\psi}'$.

It is straightforward to see that the lemma implies (A). Now we
prove that (A) implies the lemma. Suppose that there exists a site $\psi$ that
makes the lemma false. Then pick any NE $x$. There must be two NEs $x', x''$
such that $x_{0,\psi} \neq x'_{0,\psi}$ and $x_{\Sigma,\psi} \neq x''_{\Sigma,\psi}$.
If $x_{\Sigma,\psi} \neq x'_{\Sigma,\psi}$ or $x_{0,\psi} \neq x''_{0,\psi}$,
then (A) is false. If not, $x'_{\Sigma,\psi} = x_{\Sigma,\psi} \neq x''_{\Sigma,\psi}$ and
$x''_{0,\psi} = x_{0,\psi} \neq x'_{0,\psi}$, then (A) is still false using $x'$ and $x''$.
This completes the proof.

We now define 4 possible cases that would contradict (A):
\begin{align*}
--&: x_{0,\psi}' < x_{0,\psi}, x_{\Sigma,\psi}' < x_{\Sigma,\psi} \\
-+&: x_{0,\psi}' < x_{0,\psi}, x_{\Sigma,\psi}' > x_{\Sigma,\psi} \\
++&: x_{0,\psi}' > x_{0,\psi}, x_{\Sigma,\psi}' > x_{\Sigma,\psi} \\
+-&: x_{0,\psi}' > x_{0,\psi}, x_{\Sigma,\psi}' < x_{\Sigma,\psi}
\end{align*}

We prove that if any one of these 4 cases occurs, all 4 cases must all
occur (on some targets).  This results in a contradiction, because the $-+$
case implies $\theta_0' > \theta_0$ while the $+-$ case implies $\theta_0'
< \theta_0$. Hence none of the cases can occur, and our theorem holds.

Case $-+$ implying case $++$ and case $+-$ implying case $--$ can be proved
in the same way as before~\cite{Korzhyk11:Security} because the catcher is
not Bayesian.  Therefore, we focus on proving that case $--$ implies case
$-+$. (The proof of case $++$ implying case $+-$ is symmetric.)
%The only additional trick that we need is to consider a set of sites/targets
%rather just a single one.

%Assume $x_{0,\psi_1}' < x_{0,\psi_1}, x_{s,\psi_1}' < x_{s,\psi_1} (--)$. 
Assume that $\psi_1$ is of case $--$.  Let $\Delta x_{i,\psi} =
x_{i,\psi}'-x_{i,\psi} \ (i \in N^+)$ and $\Delta x_{\Sigma,\psi} =
x_{\Sigma,\psi}' - x_{\Sigma,\psi}$. Then $\Delta x_{0, \psi_1} < 0$ and
$\Delta x_{\Sigma, \psi_1} < 0$. Let $\Psi^- = \lbrace \psi ~|~ \Delta
x_{0,\psi} < 0, \Delta x_{\Sigma, \psi} \leq 0 \rbrace$.

We have $\sum_{\psi \in \Psi^-} \Delta x_{\Sigma,\psi} = \sum_{i=1}^n
\sum_{\psi \in \Psi^-} \Delta x_{i,\psi} < 0$ (because of $\psi_1$). Then
$\sum_{\psi \in \Psi^-} \Delta x_{i,\psi} < 0$ must be true for some
$i$. Let $\Psi^+ = \Psi \setminus \Psi^-$. We have $\sum_{\psi \in \Psi^+}
\Delta x_{i,\psi} > 0$ because $\sum_{\psi \in \Psi} \Delta x_{i,\psi} =
\sum_{\psi \in \Psi^-} \Delta x_{i,\psi} + \sum_{\psi \in \Psi^+} \Delta
x_{i,\psi} = 0$.  Therefore, there exist some $\psi^- \in \Psi^-$ and
$\psi^+ \in \Psi^+$ such that $\Delta x_{i,\psi^-} < 0$ and $\Delta
x_{i,\psi^+} > 0$.

Note that $\Delta x_{i,\psi^-} < 0$ implies $\theta_i' > \theta_i$
($\theta_i' \geq \mu_{i,\psi^-}' > \mu_{i,\psi^-} \geq \theta_i $ where the
weak inequalities come from $\Delta x_{i,\psi^-} < 0$ and the strict
inequality comes from $\Delta x_{0,\psi^-} < 0$). Then, we have
$\mu_{i,\psi^+}' \geq \theta_i' > \theta_i \geq \mu_{i,\psi^+}$, where the
weak inequalities come from $\Delta x_{i,\psi^+} > 0$.
From this it follows that
$x_{0,\psi^+}' < x_{0,\psi^+}$.
% (because $\mu_{i,\psi^+}' \geq \theta_i' >
%\theta_i \geq \mu_{i,\psi^+}$ where the weak inequalities come from $\Delta
%x_{i,\psi^+} > 0$). 
This then implies $\Delta x_{\Sigma,\psi^+} > 0$ (otherwise $\psi^+ \in
\Psi^-$), which allows us to conclude that case $-+$ holds for $\psi^+$.
\end{proof}
\end{fullversion}

\begin{theorem}
The Nash equilibria (NE) of a Catcher-Evader game are interchangeable.
That is, if $x$ and $x'$ are two Nash equilibrium strategy profiles,
then so is $x''$ where $x''_{0,\psi} = x_{0,\psi}$ and
$x''_{i,\psi} = x'_{i,\psi} \ (i \in N^+)$.
\end{theorem}

\begin{fullversion}
\begin{proof}
  The proof is similar to the one given by~\cite{Korzhyk11:Security} for
  interchangeability in (non-Bayesian) security games with multiple
  attacker resources.  Only the proof of Lemma 4 in that paper needs to be
  modified to our Lemma~\ref{lemma:same}. The remaining reasoning is unchanged.
\end{proof}
\end{fullversion}

\section{Computing Nash Equilibrium}

\begin{shortversion}
\begin{figure*}
\centering
	\includegraphics[width=\linewidth]{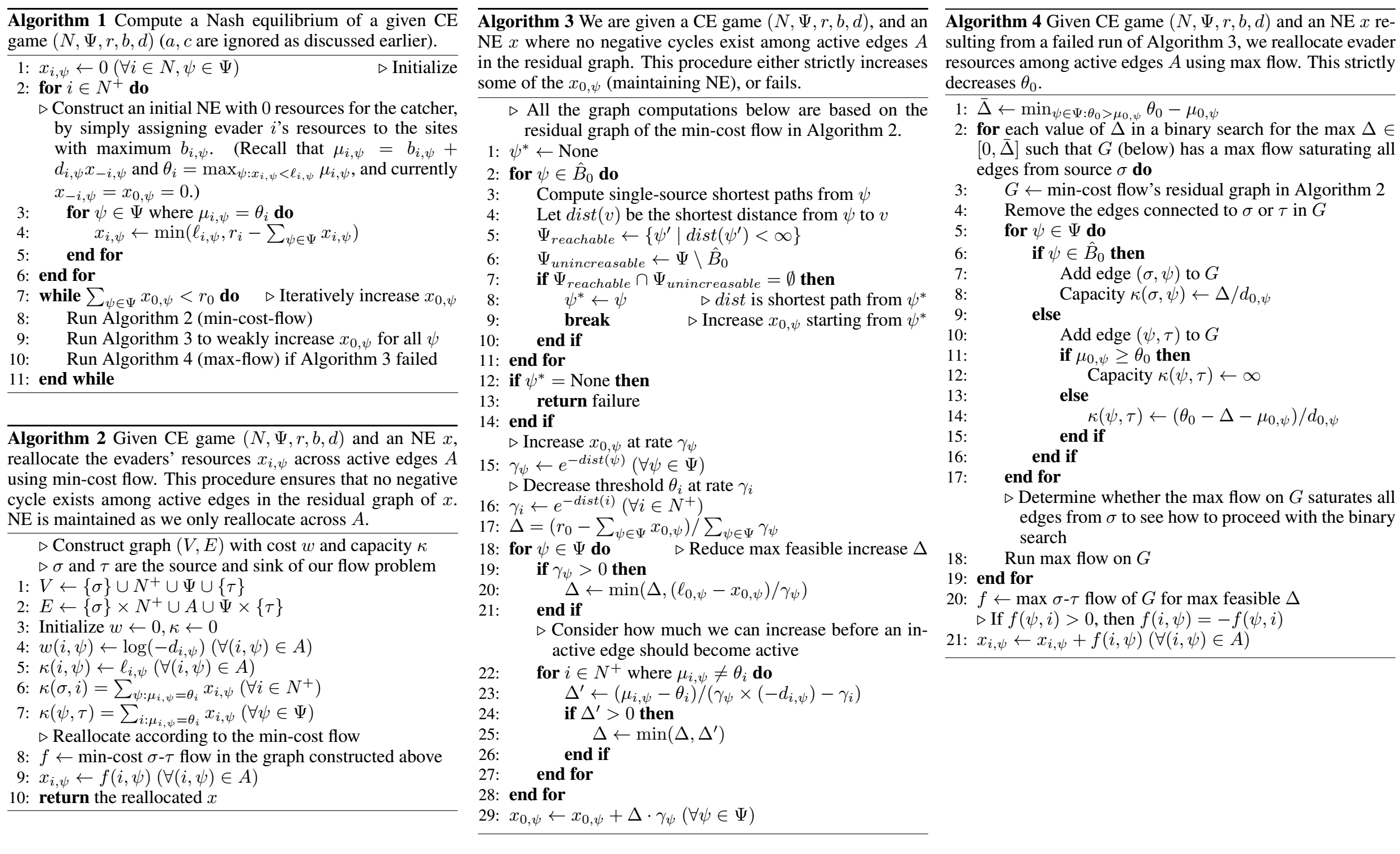}
\end{figure*}
\end{shortversion}

The interchangeability established in the previous section
provides good motivation for computing a Nash equilibrium in this domain.
In this section, we provide an algorithm for doing so.  The algorithm is
significantly more involved than earlier algorithms, notably requiring
a min-cost-flow subroutine.  This is perhaps surprising as earlier
algorithms---e.g., the one by~\cite{Korzhyk11:Security} for computing a
Nash equilibrium in non-Bayesian security games with multiple attacker
resources---do not need to do so.  However, in
the next subsection, we show it is possible to reduce the problem
of finding a minimum-cost fractional matching to our games, suggesting that
this complexity is inherent in the problem.  We have been unable to either
give a polynomial upper bound on the number of iterations of our algorithm
(each iteration takes polynomial time), or any class of instances that
results in superpolynomially many iterations.  We only give an exponential
upper bound.  However, as we will show, in experiments few iterations
suffice.

\subsection{Reducing from Min-Cost Matching}
\label{su:matching}

We show that computing an NE in CE games (even with single-resource
evaders, i.e., $\forall i \in N^+,~ \ell_{i,\psi} \geq r_i$) is as hard as
computing minimum-cost fractional\footnote{Of course, network flow problems
  have an integrality property---but not if the input is fractional, as we
  allow here.}  matchings---a common type of flow problem---suggesting that
we are unlikely to find a linear-time
algorithm.  Some of the ideas in the reduction, in particular having costs
in the graphs corresponding to the logarithms of utility change rates $d$,
will also appear in the algorithm we present later.
%hence a very simple (e.g., linear) algorithm 
%(e.g., linear time algorithm) 
%is unlikely to be found.

\begin{theorem}
  Computing a Nash equilibrium of a CE game is as hard as computing a
  minimum-cost fractional matching of a weighted bipartite
  graph. Specifically, if there is a Nash equilibrium finding algorithm
  that runs in $T(I)$ time, where $I$ is the input size of the CE game,
  then we can solve the matching problem in $T(O(I'))$ time, where $I'$ is
  the input size of the bipartite graph. So computing a NE
  is not possible in linear time unless there is a linear algorithm for
  matching.
\end{theorem}

\begin{proof}
  We reduce the matching instance to a CE game whose NE can be
  straightforwardly translated back to an optimal solution to the matching
  instance. The reduction takes linear time, resulting in the bound in the
  theorem.
%Here we show a very simple and natural reduction from best matching problems to
%Bayesian CE games. The reason that we reduce from best matching rather than min
%cost flow is because our Bayesian CE game would have put some constraints on
%the edge capacity which might make the min-cost-flow reduction complicated.

  Let the matching instance be on a bipartite graph with vertices $U = \{1,
  \ldots, n\}$ and $V$.  Each vertex $v$ has a capacity $\kappa_v$, with
  $\sum_{u \in U} \kappa_u = \sum_{v \in V} \kappa_v$.  Each edge $(u,v)$
  has a capacity $\kappa(u,v)$ and a cost $w(u,v)$.  Our goal is to saturate
  all the vertices' capacities at minimum cost.  Equivalently, this is a flow
  problem where $\sum_{u \in U} \kappa_u$ flow must be pushed across the
  bipartite graph at minimum cost.

  We construct a CE game $(N, \Psi, r, a, b, c=0, d)$ where $N=\{0, 1, 2,
  \ldots, n\}$ (so $N^+=U$), $\Psi=V$, $r_0 = 1$ and $r_u = \kappa_u$ for
  all $u \in U$, $\ell_{0,v} = 1$ and $\ell_{u,v}= \kappa(u,v)$ for all $u
  \in U$ and $v \in V$, $b_{i,v} = 0$ for all $i \in N, v \in V$, $d_{0,v}
  = 1/\kappa_v$ and $d_{u,v} = -e^{w(u,v)}$ for all $u \in U$ and $v \in V$.

  First, we note that the game has a feasible strategy for the evaders if
  and only if the matching problem has a feasible solution.  This is
  because a feasible strategy $x_{u,v}$ corresponds exactly to a feasible
  matching solution.

%  We first argue that in the Nash equilibrium, $x_{\Sigma,v}=\kappa_v$ for
%  all $v$, so that the $x_{u,v}$ correspond to a feasible solution of the
%  matching problem.  For suppose not; let $V_0 = \lbrace v ~|~ \forall v',~
%  \mu_{0,v} \leq \mu_{0,v'}\rbrace$ where $\mu_{0,v} = x_{\Sigma,v} d_{0,v}
%  = x_{\Sigma,v} / \kappa_v$.
%such
%  that $\Psi \setminus \Psi_0 \neq \emptyset$. Then we have
%  $x_{\Sigma,\psi} < x_{s,\psi'}$ for $\psi \in \Psi_0, \psi' \in \Psi
%  \setminus \Psi_0$.  Since $r_0 = 1 = r_{0,\psi} (\forall \psi \in \Psi$),
%  the catcher must put $x_{0,\psi} = 0$ resource for all $\psi \in \Psi_0$
%  (in the best response). But then, the evader must only have $x_{i,\psi}
%  > 0$ if $\psi \in \Psi_0$. This contradicts the definition of $\Psi_0$.

  Second, $x_{0,v} > 0$ must hold for all $v$. Otherwise, because
  $b_{u,v}=0$ and $d_{u,v} < 0$, all evaders will strictly prefer targets
  with $x_{0,v} = 0$; but then the catcher would not be best-responding,
  because $b_{0,v}=0$ and $d_{0,v}>0$.

%we can construct
%$\emptyset \subsetneq \Psi_0 = \{\psi ~|~ x_{0,\psi} = 0\} \subsetneq \Psi$ which
%will lead to the same contradiction above.
%The argument above also shows that $y_j > 0$ for all $j$ since when $y_j = 0$
%exists, such $j$ constitutes a $J_0$, which then leads to a contradiction.

  Finally, we show that the NE $x$ must constitute an optimal solution to
  the matching problem.  That is, if we let $W = \sum_{u \in U, v \in V}
  x_{u,v} w(u,v)$ then $W$ is the minimum cost in the matching problem.
  Suppose not; then, when interpreting $x_{u,v}$ as a flow, in the residual
  graph of that flow, a negative cycle exists. Let that cycle be $u_1
  \rightarrow v_1 \rightarrow u_2 \rightarrow v_2 \rightarrow \ldots
  \rightarrow u_m \rightarrow v_m \rightarrow u_1$ with $\sum_{1 \leq k
    \leq m} (w(u_k, v_k)-w(u_{k+1}, v_k)) < 0$, $x_{u_k, v_k} < \kappa(u_k,
  v_k)$, and $x_{u_{k+1}, v_k} > 0$ for all $1 \leq k \leq m$ (letting
  $u_{m+1} = u_1$).  Recall that $\theta_u$ is the {\em per-resource
    utility} threshold for evader $u$. So, $\mu_{u_k,v_k} = x_{0, v_k}
  d_{u_k, v_k} \leq \theta_{u_k}$ and $\mu_{u_{k+1}, v_k} = x_{0, v_k}
  d_{u_{k+1}, v_k} \geq \theta_{u_{k+1}}$.  Equivalently,\\ $x_{0, v_k}
  |d_{u_{k+1}, v_k}| \leq |\theta_{u_{k+1}}|$ and $|\theta_{u_k}| \leq
  x_{0, v_k} |d_{u_k, v_k}|$ because $d_{u,v} < 0$. It then follows
  that $$\prod_{k=1}^m x_{0, v_k} |d_{u_{k+1}, v_k}| \cdot
  |\theta_{u_k}| \leq \prod_{k=1}^m x_{0, v_k} |d_{u_k, v_k}| \cdot
  |\theta_{u_{k+1}}|$$ which implies $\prod_{k=1}^m |d_{u_{k+1}, v_k}|
  \leq \prod_{k=1}^m |d_{u_k, v_k}|$ because $x_{0,v} > 0$ for all
  $v$ and thus $|\theta_u| > 0$ for all $u \in U$.  Taking the logarithm on
  both sides, we obtain $\sum_{1\leq k \leq m} (w(u_k, v_k)-w(u_{k+1},
  v_k)) \geq 0$, contradicting the negative cycle assumption $\sum_{1\leq k
    \leq m} (w(u_k, v_k)-w(u_{k+1}, v_k)) < 0$.
%
%which implies $x_{0, \psi_k} d_{i_k, \psi_k} \leq x_{0, \psi_k} d_{i_{k+1},
%\psi_k} (\forall 1 \leq k \leq m)$. Since $x_{0, \psi} > 0 (\forall \psi \in
%\Psi)$, it implies $d_{i_k, \psi_k} \leq d_{i_{k+1}, \psi_k}$, which is
%$-e^{w_(i_k, \psi_k)} \leq -e^{w_(i_{k+1}, \psi_k)}$.  Therefore $w_{i_k,
%\psi_k}-w_{i_{k+1}, \psi_k}) \geq 0$ for all $1 \leq k \leq m$. This
%contradicts $\sum_{1\leq k \leq m} (w(i_k, j_k)-w(i_{k+1}, j_k)) < 0$.
\end{proof}

\subsection{Algorithm}

\begin{figure*}
\centering
\subfigure[Average total running time]{
	\includegraphics[trim=0 20 0 50, width=.23\linewidth]{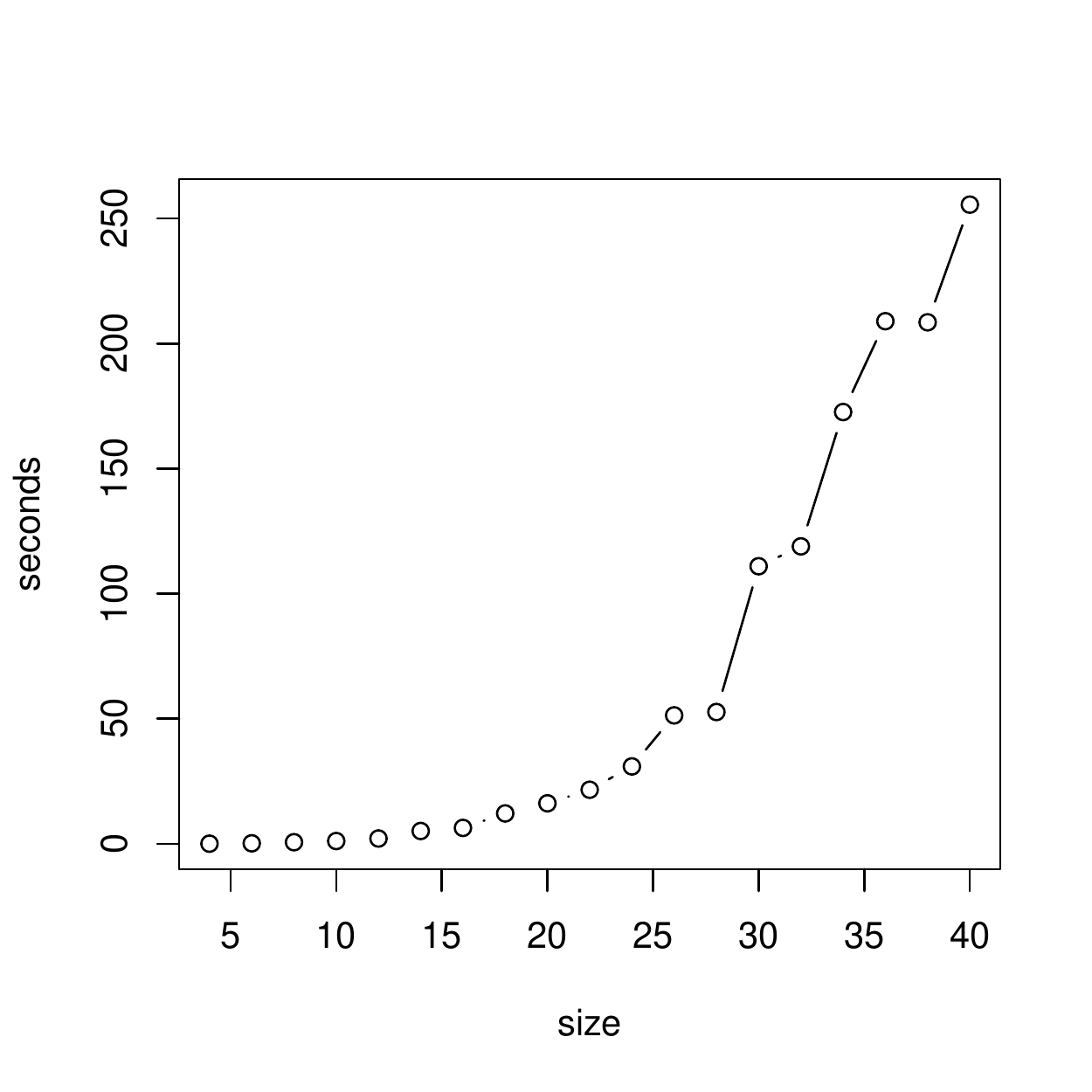}
	\label{fig:time}
}
\subfigure[Average time per iteration]{
	\includegraphics[trim=0 20 0 50, width=.23\linewidth]{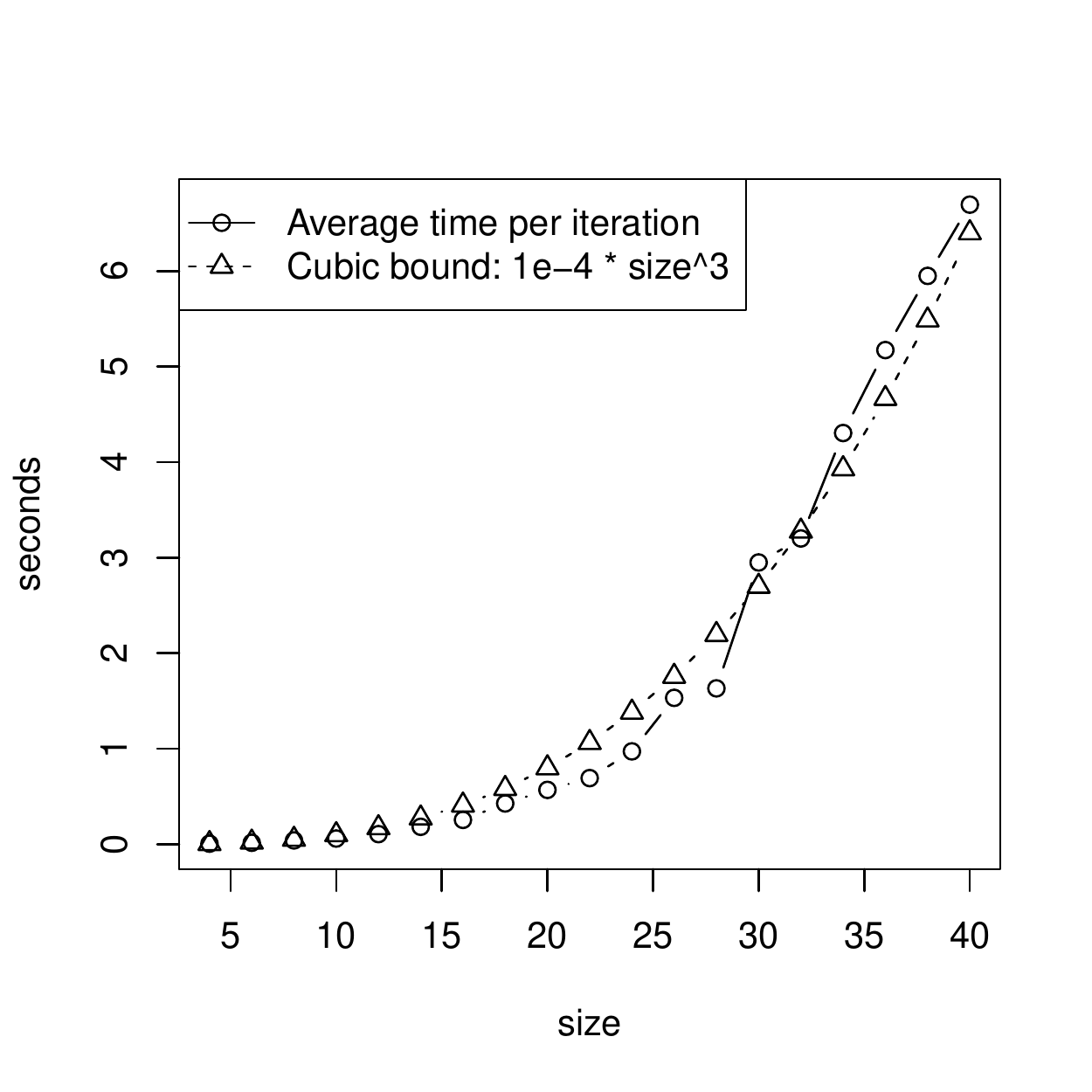}
	\label{fig:iter_time}
}
\subfigure[Number of iterations]{
	\includegraphics[trim=0 20 0 50, width=.23\linewidth]{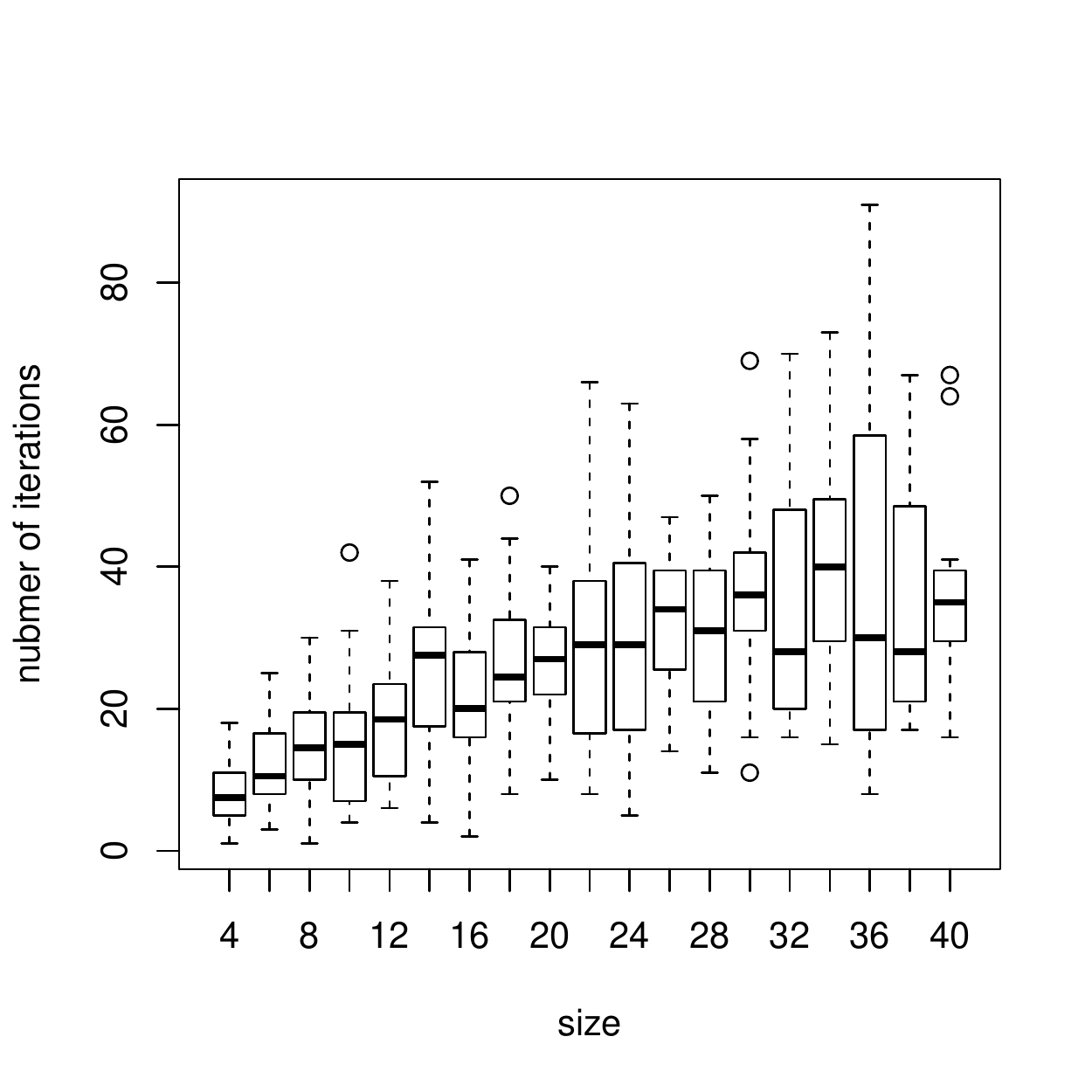}
	\label{fig:exp1}
}
\subfigure[Size of normal form]{
	\includegraphics[trim=0 20 0 50, width=.23\linewidth]{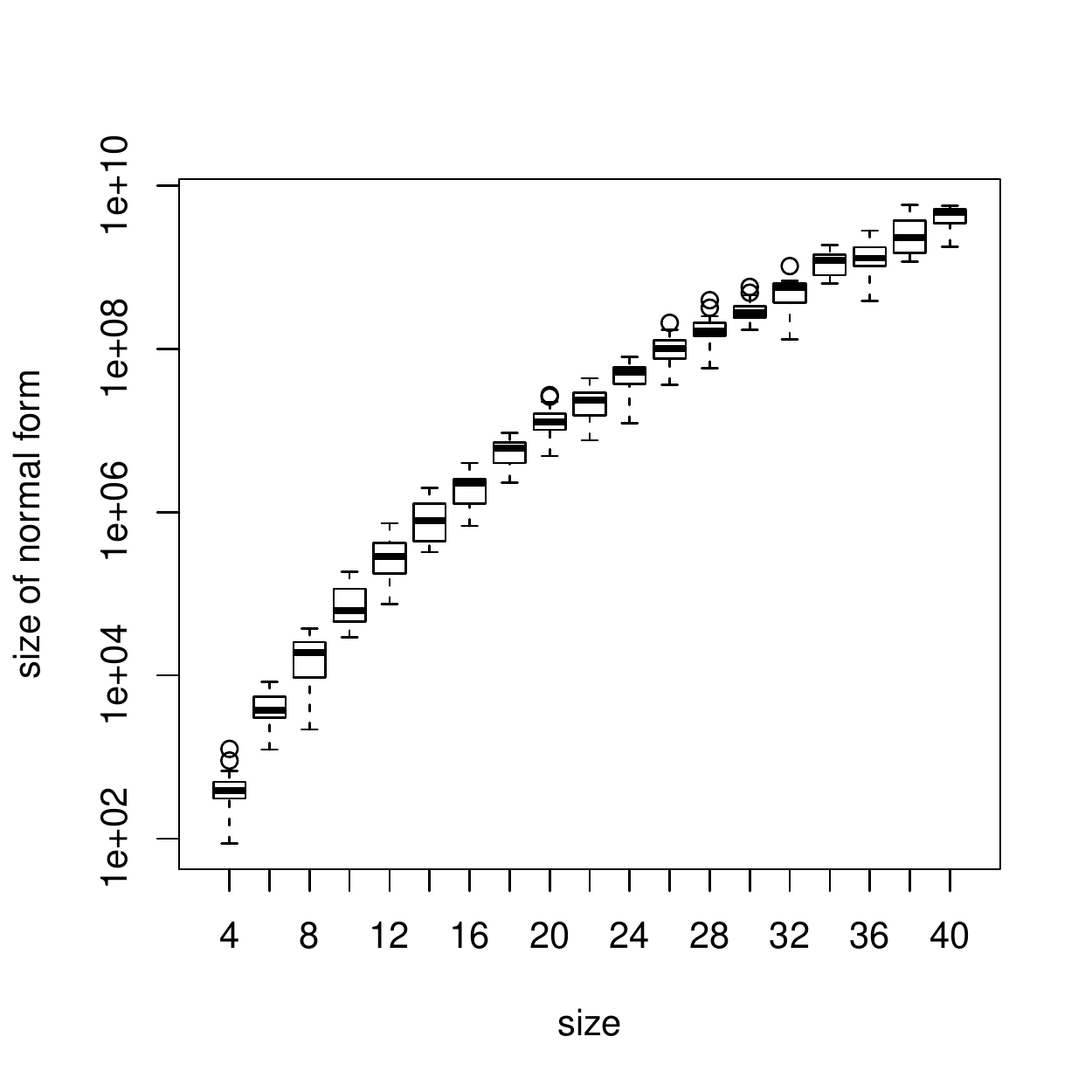}
	\label{fig:exp2}
}
\caption{
Performance of Algorithm 1, and the size
of the normal form, for randomly generated CE games.
}
\end{figure*}

We now present our algorithm.  The algorithm works by initializing the
catcher's resource amount to $0$ and gradually increasing it to $r_0$,
maintaining the equilibrium throughout.  The same high-level approach was
used by an earlier paper~\cite{Korzhyk11:Security} for the case of a single attacker type
(evader) with multiple resources, obtaining an efficient algorithm there.
However, the case with multiple evaders is significantly more involved.
The (polynomial-time) algorithm given in~\cite{Korzhyk11:Security} did not
require anything like a network-flow subroutine, whereas the reduction in
section~\ref{su:matching} suggests that this is necessary when we have
multiple evaders.  Our algorithm also incorporates ideas used in
the context of test games~\cite{Li13:Game}, specifically the binary
search and max-flow techniques used there.

We first introduce some notation. Let $B_i = \{\psi ~|~ \mu_{i,\psi}
= \theta_i\}$ be the {\em boundary sites} of player $i$. Let $B_i^+ =
\{\psi \in B_i \land x_{i,\psi} > 0\}$ be evader $i \ (i \in N^+)$'s {\em
  positive} boundary sites, whose resource amount can be reduced. Let $\hat
B_0 = \{\psi \in B_0 \land x_{0,\psi} < \ell_{0,\psi}\}$ be the catcher's {\em
  open} boundary sites, to which more resources could be assigned. Let the
{\em active edges} be $A = \{(i, \psi) ~|~ i \in N^+, \psi \in \Psi, \psi
\in B_i\}$.

The main algorithm is Algorithm 1. After initializing, the
algorithm repeatedly loops through Algorithms 2, 3, and 4, which
together provably (eventually) increase the catcher's (allocated)
resource amount while maintaining equilibrium.
Algorithm 2 ensures that a ``no negative cycle'' property holds by
solving a min-cost flow problem (since the residual flow
of a min-cost flow cannot have a negative cycle).
Here, the relationship between the evaders' best responses and the
min-cost flow's ``no
negative cycle'' property is similar to the reduction from min-cost
matching that we gave earlier.
Given that no negative cycle
remains, Algorithm 3 then attempts to increase the catcher's resource
amount---that is, for each $\psi$ it attempts to
 increase $x_{0,\psi}$---without
breaking the evaders' best-response conditions.
However,
Algorithm 3 can still fail to increase the catcher's resource amount
even without negative
cycles.  If so, we call Algorithm 4, which will either allow the next
run of Algorithm
3 to strictly increase the catcher's resource amount, or change the
open boundary sites $\hat{B}_0$ (which provably
cannot happen too often).
Specifically, if Algorithm 3 failed to increase the catcher's resource
amount, we have to reroute evaders' resources
among their best-response sites, in a way that strictly decreases the
catcher's utility threshold $\theta_0$. Such rerouting must also
maintain the catcher's best-response condition.  For this we use max-flow
and binary search: first, we binary search on $\Delta$ , the decrease
in $\theta_0$; then, we calculate each edge's rerouting capacity using
$\Delta$, and see whether a max-flow can saturate
all capacities, thereby maintaining the best-response condition.

\begin{fullversion}
\begin{breakablealgorithm}
\caption{Compute a Nash equilibrium of a given CE game $(N, \Psi, r, b, d)$
($a, c$ are ignored as discussed earlier).}
\label{algorithm:main}
\begin{algorithmic}[1]

\State $x_{i,\psi} \gets 0 \ (\forall i \in N, \psi \in \Psi)$ \Comment{Initialize}
\For {$i \in N^+$}
\LineCommentCont{Construct an initial NE with $0$ resources for the
  catcher, by
	simply assigning evader $i$'s resources to the sites with maximum $b_{i,\psi}$.
	(Recall that $\mu_{i,\psi} = b_{i,\psi}+d_{i,\psi}x_{-i,\psi}$ and
	$\theta_i = \max_{\psi: x_{i,\psi} < \ell_{i,\psi}} \mu_{i,\psi}$,
        and currently
	$x_{-i,\psi} = x_{0,\psi} = 0$.)
}
%	\State $\Psi^*_i \gets \{\psi ~|~ (\forall \psi' \in \Psi) \mu_{i,\psi} \geq \mu_{i,\psi'}\}$
%	\LineCommentCont{Recall that $\mu_{i,\psi} = b_{i,\psi}+d_{i,\psi}x_{-i,\psi}$}
%	\For {$\psi \in \Psi^*_i$}
	\For {$\psi \in \Psi$ where $\mu_{i,\psi} = \theta_i$}
		\State $x_{i,\psi} \gets \min(\ell_{i,\psi}, r_i-\sum_{\psi \in \Psi} x_{i,\psi})$
	\EndFor
\EndFor

\While {$\sum_{\psi \in \Psi} x_{0, \psi} < r_0$} \Comment{Iteratively increase $x_{0,\psi}$}
	\State Run Algorithm~\ref{algorithm:min-cost-flow} (min-cost-flow)
	\State Run Algorithm~\ref{algorithm:increase} to weakly increase
        $x_{0,\psi}$ for all $\psi$
	\State Run Algorithm~\ref{algorithm:max-flow} (max-flow) if Algorithm~\ref{algorithm:increase} failed
\EndWhile

\end{algorithmic}
\end{breakablealgorithm}

\begin{breakablealgorithm}
\caption{Given CE game $(N, \Psi, r, b, d)$ and an NE $x$, reallocate the evaders'
resources $x_{i,\psi}$ across active edges $A$ using min-cost flow. This
procedure ensures that no negative cycle exists among active edges in the
residual graph of $x$. NE is maintained as we only reallocate across $A$.}
\label{algorithm:min-cost-flow}

\begin{algorithmic}[1]
\LineComment{Construct graph $(V,E)$ with cost $w$ and capacity $\kappa$}
\LineComment{$\sigma$ and $\tau$ are the source and sink of our flow problem}
%\State $N^+ \gets N \setminus \{0\}$ (set of evaders)
\State $V \gets \{\sigma\} \cup N^+ \cup \Psi \cup \{\tau\}$
\State $E \gets \{\sigma\} \times N^+ \cup A \cup \Psi \times \{\tau\}$
\State Initialize $w \gets 0, \kappa \gets 0$
\State $w(i,\psi) \gets \log (-d_{i,\psi}) ~(\forall (i,\psi) \in A)$
%$~(\forall i \in N^+, \psi \in \Psi)$
%\State $w(s,i) \gets 0, w(\psi,t) \gets 0 ~(\forall i \in N^+, \psi \in \Psi)$
\State $\kappa(i,\psi) \gets \ell_{i,\psi} ~(\forall (i,\psi) \in A)$
 %	$~(\forall i \in N^+, \psi \in \Psi: \mu_{i,\psi} = \theta_i)$
\State $\kappa(\sigma, i) = \sum_{\psi: \mu_{i,\psi} = \theta_i} x_{i,\psi} ~(\forall i \in N^+)$
\State $\kappa(\psi, \tau) = \sum_{i: \mu_{i,\psi} = \theta_i} x_{i,\psi} ~(\forall \psi \in \Psi)$
\LineComment{Reallocate according to the min-cost flow}
\State $f \gets $ min-cost $\sigma$-$\tau$ flow in the graph constructed above
\State $x_{i,\psi} \gets f(i,\psi) ~(\forall (i, \psi) \in A)$
\State \Return the reallocated $x$
\end{algorithmic}

\end{breakablealgorithm}

\begin{breakablealgorithm}
\caption{We are given a CE game $(N, \Psi, r, b, d)$, and an NE $x$ where no
negative cycles exist among active edges $A$
%$A = \{(i, \psi) ~|~ \theta_i = \mu_{i,\psi}\}$ 
in the residual graph. This procedure either strictly increases some
of the 
$x_{0,\psi}$ (maintaining NE), or fails.}
\label{algorithm:increase}

\begin{algorithmic}[1]
\LineCommentCont{All the graph computations below are based on the
	residual graph of the min-cost flow in
	Algorithm~\ref{algorithm:min-cost-flow}.}
\State $\psi^* \gets $ None
\For {$\psi \in \hat B_0$}
	\State Compute single-source shortest paths from $\psi$
	\State Let $dist(v)$ be the shortest distance from $\psi$ to $v$
	\State $\Psi_{reachable} \gets \{\psi ~|~ dist(\psi) < \infty \}$
	\State $\Psi_{unincreasable} \gets \Psi \setminus \hat B_0$
	\If {$\Psi_{reachable} \cap \Psi_{unincreasable} = \emptyset$}
		\State $\psi^* \gets \psi$ \Comment{$dist$ is shortest path from $\psi^*$}
		\State \textbf{break} \Comment{Increase $x_{0,\psi}$ starting from $\psi^*$}
	\EndIf
%	\State Let $dist(v)$ be the shortest distance from $s$ to $v$ in the residual
%		graph defined in Algorithm~\ref{algorithm:min-cost-flow}, with respect
%		to its min cost flow $x$ (i.e., $f$) and edge cost/weight $w$.
\EndFor

\If {$\psi^* = $ None}
	\State \Return failure
\EndIf

\LineComment{Increase $x_{0,\psi}$ at rate $\gamma_\psi$}
\State $\gamma_\psi \gets e^{-dist(\psi)} ~(\forall \psi \in \Psi)$
\LineComment{Decrease threshold $\theta_{i}$ at rate $\gamma_i$}
\State $\gamma_i \gets e^{-dist(i)} ~(\forall i \in N^+)$

\State $\Delta = (r_0 - \sum_{\psi \in \Psi} x_{0,\psi})
	/ \sum_{\psi \in \Psi} \gamma_\psi$
\For {$\psi \in \Psi$} \Comment{Reduce max feasible increase $\Delta$}
	% The following 'if' can be dropped if we allow x/0 = infty
	\If {$\gamma_\psi > 0$}%\Comment{reachable, i.e., $dist(\psi) < \infty$}
		\State $\Delta \gets \min(\Delta, 
				(\ell_{0,\psi}-x_{0,\psi})/\gamma_\psi)$ 
				\label{line:saturate-constraint}
	\EndIf
\LineCommentCont{Consider how much we can increase before an inactive edge
  should become active}
	\For {$i \in N^+$ where $\mu_{i,\psi} \neq \theta_i$}
		\State $\Delta' \gets (\mu_{i,\psi}-\theta_i)
				/(\gamma_\psi \times (-d_{i,\psi})-\gamma_i)$
				\label{line:new-active-constraint}
		\If {$\Delta' > 0$} 	
			\State $\Delta \gets \min(\Delta, \Delta')$ \label{line:increase-delta}
		\EndIf
	\EndFor
\EndFor

\State $x_{0,\psi} \gets x_{0,\psi} + 
		\Delta \cdot \gamma_{\psi} ~(\forall \psi \in \Psi)$
\end{algorithmic}

\end{breakablealgorithm}

\begin{breakablealgorithm}
\caption{Given CE game $(N, \Psi, r, b, d)$ and an NE $x$ resulting from  a
  failed run of Algorithm~\ref{algorithm:increase}, we reallocate evader resources among
active edges $A$
%$\{(i, \psi) ~|~ \theta_i = \mu_{i,\psi}\}$
 using max flow.  This
strictly decreases $\theta_0$.}
\label{algorithm:max-flow}

\begin{algorithmic}[1]
\State $\bar \Delta \gets \min_{\psi \in \Psi: \theta_0 > \mu_{0,\psi}} 
		\theta_0 - \mu_{0,\psi}$
\For{each value of $\Delta$ in a binary search for the  max $\Delta \in [0,\bar \Delta]$ such that 
 $G$ (below) has a max flow saturating all edges from source $\sigma$}

%\State{Binary search for the  max $\Delta \in [0,\bar \Delta]$ such that $G$ (below) has a max flow saturating all edges from source $s$.}
\State $G \gets $ min-cost flow's residual graph in 
		Algorithm~\ref{algorithm:min-cost-flow}
\State Remove the edges connected to  $\sigma$ or  $\tau$ in $G$
\For {$\psi \in \Psi$}
	\If {$\psi \in \hat{B}_0$}
%{$\mu_{0,\psi} = \theta_0 \land x_{0,\psi} < r_{0,\psi}$}
		\State Add edge $(\sigma, \psi)$ to $G$
		\State Capacity $\kappa(\sigma,\psi) \gets \Delta/d_{0,\psi}$
	\Else
		\State Add edge $(\psi, \tau)$ to $G$
		\If {$\mu_{0,\psi} \geq \theta_0$}
			\State Capacity $\kappa(\psi,\tau) \gets \infty$
		\Else
			\State $\kappa(\psi,\tau) \gets 
				(\theta_0-\Delta-\mu_{0,\psi})/d_{0,\psi}$ \label{line:decrease-kappa}
		\EndIf
	\EndIf
\EndFor
\LineCommentCont{Determine whether the max flow on $G$ saturates all edges
  from  $\sigma$ to see how to proceed with the binary search} 
\State{Run max flow on $G$}
\EndFor
\State $f \gets $ max $\sigma$-$\tau$ flow of $G$ for max feasible $\Delta$
\LineComment{If $f(\psi,i) > 0$, then $f(i,\psi)=-f(\psi,i)$}
\State $x_{i,\psi} \gets x_{i,\psi} + f(i,\psi) ~(\forall (i,\psi) \in A)$
\end{algorithmic}
\end{breakablealgorithm}
\end{fullversion}

\begin{lemma}
\label{lemma:no-change}
Algorithm 2 maintains Nash equilibrium without
changing $\mu_{i,\psi}$ or $\theta_i$ for any $i \in N^+$ and $\psi \in \Psi$.
As a result, the active edges $A$ are also unchanged.
\end{lemma}
\begin{proof}
  Algorithm 2 only reallocates $x_{i,\psi}$
  among $A$, hence the evaders necessarily continue to best respond.  Both
  the original flow and the min-cost flow are required to saturate all
  edges $(\psi, \tau)$. Hence $x_{\Sigma,\psi}$ is unchanged for all $\psi
  \in \Psi$ and the catcher necessarily continues to best respond.  Each
  evader $i$'s $\mu_{i,\psi}$ is clearly unchanged as $x_{0,\psi}$ is
  untouched by Algorithm 2.  By the definition
  of $\theta_i = \min_{\psi: x_{i, \psi} > 0} \mu_{i,\psi}$, the set of
  positive boundary sites $B_i^+$ must be non-empty, which means
  $\sum_{\psi \in B_i} x_{i,\psi} >0$.  So no matter how we reallocate,
  some $\psi \in B_i$ must remain positive.  Hence, $\theta_i$ is unchanged
  because the $\mu_{i,\psi}$ are unchanged.
\end{proof}

\begin{lemma}
\label{lemma:rate}
In Algorithm 3, the evaders' thresholds $\theta_i \
(i \in N^+)$ decrease at rate $\gamma_i$. That is, the algorithm decreases
$\theta_i$ by $\Delta \cdot \gamma_i$.
\end{lemma}
\begin{proof}
  Throughout this proof, we require $i \in N^+$.  We also denote by $\mu^0,
  \theta^0, A^0, B_i^0, B_i^{+0}$ the values of $\mu, \theta, A, B_i,
  B_i^+$ before increasing $x_{0,\psi}$.
%First, we prove that the decrease rate of $\theta_i$ is not affected by $\psi$
%where $\mu^0_{i,\psi} \neq \theta^0_i$ (or equivalently, $(i,\psi) \notin
%A^0$). 

Because by definition, $\theta_i = \min_{\psi: x_{i, \psi} > 0} \mu_{i,\psi}$,
clearly $\theta_i$ is unaffected by sites $\psi$ for which  $\mu^0_{i,\psi} < \theta^0_i$,
because $x_{i,\psi}=0$ by the
best-response property.
For $\psi$ where $\mu^0_{i,\psi} > \theta^0_i$, it will not affect $\theta_i$
for small enough $\Delta \leq \min_{\psi: \mu^0_{i,\psi} > \theta^0_i}
(\mu^0_{i,\psi}-\theta^0_i)/\gamma_\psi$.

Therefore, for small enough $\Delta$, the threshold $\theta_i$ is only
affected by sites $\psi$ where $\mu^0_{i,\psi} = \theta^0_i$, or
equivalently $(i,\psi) \in A^0$.  
%Let $B^{+0}_i = \{\psi ~|~ x_{i,\psi} > 0
%\land \mu^0_{i,\psi} = \theta^0_i\}$ be the subset of $i$'s initial
%boundary sites $B_i^0$ that have positive $x_{i,\psi}$.  
Thus, for small enough $\Delta$:
%\begin{align*}
%\theta_i &= \min_{\psi \in B^{+0}_i} \mu^0_{i,\psi} + \Delta \cdot
%		\gamma_\psi \cdot d_{i,\psi}\\ 
%	&= \theta^0_i + \Delta \cdot \min_{\psi \in B^{+0}_i}
%		\gamma_\psi \cdot d_{i,\psi}\\
%%	&= \theta^0_i + \Delta \cdot \min_{\psi \in B^{+0}_i}
%%		e^{-dist(\psi)} \cdot d_{i,\psi} \\
%	&= \theta^0_i + \Delta \cdot \min_{\psi \in B^{+0}_i}
%		e^{-dist(\psi)} \cdot (-e^{w(i,\psi)}) \\
%	&= \theta^0_i - \Delta \cdot \max_{\psi \in B^{+0}_i}
%		e^{-dist(\psi)+w(i,\psi)} \\
%	&= \theta^0_i - \Delta \cdot 
%		e^{-\min_{\psi \in B^{+0}_i} dist(\psi)-w(i,\psi)} 
%\end{align*} %yuqian: save space
$
\theta_i = \min_{\psi \in B^{+0}_i} \mu^0_{i,\psi} + \Delta \cdot
		\gamma_\psi \cdot d_{i,\psi}
	= \theta^0_i + \Delta \cdot \min_{\psi \in B^{+0}_i}
		\gamma_\psi \cdot d_{i,\psi}
%	&= \theta^0_i + \Delta \cdot \min_{\psi \in B^{+0}_i}
%		e^{-dist(\psi)} \cdot d_{i,\psi} \\
	= \theta^0_i + \Delta \cdot \min_{\psi \in B^{+0}_i}
		e^{-dist(\psi)} \cdot (-e^{w(i,\psi)}) 
	= \theta^0_i - \Delta \cdot \max_{\psi \in B^{+0}_i}
		e^{-dist(\psi)+w(i,\psi)} 
	= \theta^0_i - \Delta \cdot 
		e^{-\min_{\psi \in B^{+0}_i} dist(\psi)-w(i,\psi)} 
$
Note that $\psi \in B^{+0}_i$ means that a backward edge $(\psi, i)$
exists in the residual graph with weight $-w(i,\psi)$. Those are
the only edges that lead to $i$, hence $dist(i) = \min_{\psi \in B^{+0}_i}
dist(\psi) - w(i,\psi)$. Therefore, $\theta_i = \theta^0_i - \Delta \cdot
e^{-dist(i)} = \theta^0_i - \Delta \cdot \gamma_i$ for small enough
$\Delta$.
All that remains to show is that $\Delta$ is in fact small enough.
This is so because line 25 of Algorithm 3
ensures that we stop decreasing before any $\psi \notin B^{+0}_i$ can
affect $\theta_i$.  Hence the lemma holds.
%Note that $\mu, \theta, A$ are unchanged throughout
%Algorithm~\ref{algorithm:min-cost-flow} (due to Lemma~\ref{lemma:no-change}),
%so $A^0$ (and corresponding reverse edges) are the possible edges between $N^{+0}$
%and $\Psi$ in the residual graph.
%
%, as line~\ref{line:increase-delta} ensures that we will stop
%decreasing before those $\psi$ can affect $\theta_i$.
%
%Therefore, $\theta_i = \min_{\psi: x_{i,\psi} > 0} \mu_{i,\psi} $
\end{proof}

\begin{lemma}
\label{lemma:increase-ne}
After Algorithm 2, if
Algorithm 3 successfully increases $x_{0,\psi}$, it
maintains Nash equilibrium.
\end{lemma}
\begin{proof}
  The catcher's strategy remains a best response because
  Algorithm 3 does not change any evader's strategy
  and the catcher only increases $x_{0,\psi}$ for which $\mu_{0,\psi} =
  \theta_0$.  Thus we only have to check whether each evader's strategy
  remains a best response.

By Lemma~\ref{lemma:rate} and the notation $\mu^0, \theta^0, A^0$ defined in its proof,
evaders are best-responding if and only if:
\begin{align*}
	&(\forall i \in N^+, \psi \in \Psi: x_{i,\psi} > 0)\\
		&~~\mu_{i,\psi} = \mu^0_{i,\psi}+\Delta \cdot \gamma_\psi \cdot d_{i,\psi}
		\geq \theta_i = \theta^0_i - \Delta \cdot \gamma_i\\
	&(\forall i \in N^+, \psi \in \Psi: x_{i,\psi} < \ell_{i,\psi})\\
		&~~\mu_{i,\psi} = \mu^0_{i,\psi}+\Delta \cdot \gamma_\psi \cdot d_{i,\psi}
		\leq \theta_i = \theta^0_i - \Delta \cdot \gamma_i
\end{align*}

For $(i,\psi) \notin A^0$, line 25 of
Algorithm 3 maintains the conditions above.  Now
consider $(i,\psi) \in A^0$. There, we have $\mu^0_{i,\psi} = \theta^0_i$, so
we only need to check
\begin{align*}
	&(\forall i \in N^+, \psi \in \Psi: x_{i,\psi} > 0)
		~~\Delta \cdot \gamma_\psi \cdot d_{i,\psi}
		\geq - \Delta \cdot \gamma_i\\
	&(\forall i \in N^+, \psi \in \Psi: x_{i,\psi} < \ell_{i,\psi})
		~~\Delta \cdot \gamma_\psi \cdot d_{i,\psi}
		\leq - \Delta \cdot \gamma_i
\end{align*}

If $x_{i,\psi} > 0$, a backward edge $(\psi, i)$
with weight $-w(i,\psi)$ exists in the residual graph. Hence $dist(i) \leq
dist(\psi) - w(i,\psi)$, which implies $-e^{-(dist(\psi) - w(i,\psi))} \geq
-e^{-dist(i)}$. That is, $\gamma_\psi d_{i,\psi} \geq -\gamma_i$, and therefore
$\Delta \cdot \gamma_\psi \cdot d_{i,\psi} \geq - \Delta \cdot \gamma_i$ because
$\Delta \geq 0$.

If $x_{i,\psi} < \ell_{i,\psi}$, a forward edge $(i, \psi)$
with weight $w(i,\psi)$ exists in the residual graph. Hence $dist(\psi) \leq
dist(i) + w(i,\psi)$, which implies $-e^{-(dist(\psi) - w(i,\psi))} \leq
-e^{-dist(i)}$. That is. $\gamma_\psi d_{i,\psi} \leq -\gamma_i$, and therefore
$\Delta \cdot \gamma_\psi \cdot d_{i,\psi} \leq - \Delta \cdot \gamma_i$ because
$\Delta \geq 0$, completing the proof.
%
%This completes our proof.
%
%We first argue that $\theta_i (i \in N^+)$ does decrease with rate $\gamma_i$.
%Since $x_{i,\psi} > 0$
%
%We first check those $i \in N^+$ where $dist(i) = \infty$ (unreachable from
%$\psi^*$ in the residual graph).  For such $i$, we must have $dist(\psi) =
%\infty$ for all $\psi$ where $x_{i,\psi} > 0$ and $\mu_{i,\psi} = \theta_i$
%(edge $(\psi, i)$ exists in the residual graph). That is, for all $\psi$ where
%$x_{i,\psi} > 0$, the increase rate $\gamma_{\psi} = e^{-dist(\psi)} = 0$.
%Hence the catcher only increases $x_{0,\psi}$ where $x_{i,\psi}=0$, which
%maintains the best response for those evaders.
%
%
%For each pair $(i,\psi)$ where $i \in N^+$, $\psi \in \Psi$, and $\mu_{i,\psi} =
%\theta_i$ at the beginning of Algorithm~\ref{algorithm:increase}, we prove
%$\gamma_i \leq \gamma_\psi \cdot (-d_{i,\psi})$ if $x_{i,\psi} < 1$, and
%$\gamma_i \geq \gamma_\psi \cdot (-d_{i,\psi})$ if $x_{i,\psi} > 0$.
\end{proof}

\begin{lemma}
\label{lemma:decrease-catcher}
If Algorithm 3 fails, then
Algorithm 4 strictly decreases $\theta_0$ while
maintaining Nash equilibrium.
\end{lemma}
\begin{proof}
  If Algorithm 3 fails, then for each $\psi \in \hat
  B_0$, there must be another site $\psi' \in \Psi_{reachable} \cap
  \Psi_{unincreasable}$. That is, for each site $\psi \in \hat B^0$ to
  which the catcher could increase resource assignment, there is a path in
  the residual graph that goes from $\psi$ to some site $\psi' \in
  \Psi_{unincreasable}$ to which the catcher cannot increase resource
  assignment.

  That site $\psi'$ is, in contrast, a good site for evaders: if they put
  more resources there, the catcher cannot penalize them (because the
  catcher cannot increase its resources there). Formally, there are two
  cases for $\psi'$: 1) $x_{0,\psi'} = \ell_{0,\psi'}$; 2) $x_{0,\psi'} <
  \theta_0$. In the former case, evaders can increase their resource
  assignment there as much as possible because the catcher has hit the
  limit of what it can assign there.  In the latter case, evaders can
  increase until $\mu_{0,\psi'}$ meets $\theta_0$.

  Therefore, for each $\psi \in \hat B_0$, evaders can move a positive
  amount of resource from that site $\psi$ to some $\psi' \notin B_0$ using
  the corresponding residual graph path. The evaders continue to
  best-respond because the residual graph only includes active edges
  $A$. The catcher's best-response condition is maintained by decreasing
  $\mu_{0,\psi}$ by the same positive amount $\Delta$ (the number found by the
  binary search in Algorithm 4) for each $\psi \in
  \hat B_0$ (saturating all edges leaving  $\sigma$), and not letting
  $\mu_{0,\psi} \  (\psi \in \hat B_0)$  decrease below $\mu_{0,\psi'} \
  (\psi' \notin \hat B_0)$ (line 14 of
  Algorithm 4). 
%This also ensures that any site  that leaves $\hat B_0$
% during the execution of the algorithm has become saturated with catcher resources.  Hence, the catcher's
%  threshold $\theta_0$ is decreased as $\theta_0 = \max_{\psi \in \hat B_0}
%  \mu_{0,\psi}$.
  Because $\mu_{0,\psi}$ has strictly decreased for each $\psi \in \hat
  B_0$ and has not become lower than any $\mu_{0,\psi'} \ (\psi' \notin
  \hat B_0)$, $\theta_0$ must have strictly decreased (by $\Delta$).
%
%
%In summary, there is a positive $\Delta$ where a max flow could saturate all
%edges leaving from the source $\sigma$. Hence the maximum $\Delta$ we got from
%binary search is positive. We then construct a new Nash equilibrium
%by moving evaders' resource according to the max flow (because all edges are
%active). The new strategy profile decreases $\theta_0$ by $\Delta$.
\end{proof}

\begin{lemma}
\label{lemma:fail-success}
After Algorithm 4, either a new site $\psi$ that
previously had $\mu_{0,\psi} < \theta_0$ now has $\mu_{0,\psi} = \theta_0$,
or the next run of Algorithm 3 will be successful.
\end{lemma}
\begin{proof}
  Suppose that the next run of Algorithm 3 fails.
  Then for each $\psi \in \hat B_0$, a path exists in the residual graph
  (after the run of Algorithm 4) from $\psi$ to some
  $\psi' \notin \hat B_0$, as argued in the proof of
  Lemma~\ref{lemma:decrease-catcher}. Suppose, for the sake of
  contradiction, that none of the edges entering $\tau$ were saturated
  during the run of Algorithm 4 (for the value of
  $\Delta$ resulting from the binary search), i.e., $(\forall
  \psi' \notin \hat B_0), f(\psi, \tau) < \kappa(\psi, \tau)$.  Then, in
  Algorithm 4, we could have increased $\Delta$
  further, resulting in a contradiction.
%and saturated all edges leaving source $\sigma$, if all the
%  edges entering sink $\tau$ are not saturated: $(\forall \psi' \notin \hat
%  B_0), f(\psi, \tau) < \kappa(\psi, \tau)$. T
  Therefore, there exists at least one $\psi' \notin \hat B_0$ for which
  the run of Algorithm 4 made it the case that
  $f(\psi', \tau) = \kappa(\psi', \tau)$.
%for $\Delta$ being maximum. 
  For that $\psi'$, the total amount of evader resource $x_{\Sigma,\psi'}$
  was increased by $\kappa(\psi', \tau) =
  (\theta^0_0-\Delta-\mu^0_{0,\psi'})/d_{0,\psi'}$ (where superscript $0$
  denotes the value prior to the run of Algorithm
  4).  It follows that $\mu_{0,\psi} =
  \mu^0_{0,\psi} + \theta^0_0 - \Delta - \mu^0_{0,\psi} = \theta^0_0-\Delta
  = \theta_0$, while $\mu^0_{0,\psi} < \theta^0_0$.
\end{proof}

\begin{lemma}
\label{lemma:enter-once}
Each site $\psi$ enters $\hat B_0$ at most once; hence, $\hat B_0$ changes at
most $2|\Psi|$ times.
\end{lemma}
\begin{proof}
  Only Algorithm 4 can change $\mu_{0,\psi}$ or
  $\theta_0$. That algorithm ensures that $\mu_{0,\psi}$ decreases at the
  same rate for all $\psi \in \hat B_0$, and stops decreasing if a new
  $\psi'$ 
  %for which originally % yuqian: save space
  $\mu_{0,\psi'} < \theta_0$ enters $\hat
  B_0$.  So a site $\psi$ can only leave $\hat B_0$ by being saturated
  ($x_{0,\psi} = \ell_{0,\psi}$). Since we never decrease $x_{0,\psi}$
  during the algorithm, saturated sites $\psi$ can never enter $\hat B_0$
  again.
\end{proof}

\begin{lemma}
\label{lemma:bound}
Algorithm 3 runs successfully at most $2 |\Psi| \cdot 3^{n|\Psi|}$
times.
\end{lemma}

\begin{proof}
By Lemma~\ref{lemma:enter-once}, we only have to argue that there are at most
$3^{n|\Psi|}$ successful runs before $\hat B_0$ changes. Now we assume that
$\hat B_0$ remains unchanged and check how many runs there can be.

We classify an edge $(i, \psi)$ where $i \in N^+, \psi \in \Psi$ into 3 cases:
either (1) superior $\mu_{i,\psi} > \theta_i$ (above threshold), or (2) inferior
$\mu_{i,\psi} < \theta_i$ (below threshold), or (3) active $\mu_{i,\psi} =
\theta_i$ (on threshold). 

%Note that superior edges must keep saturated
%$x_{i,\psi} = \ell_{i,\psi}$, inferior edges must keep empty $x_{i,\psi} = 0$,
%and active edges constitute graph $G$ where we do min-cost flow or max-flow to
%determine what $x_{i,\psi}$ is.

Let $\phi$ be a vector of length $n|\Psi|$ where each element $\phi_e \in \{1,
2, 3\}$ denotes edge $e$'s case number. We will show that $\phi$ changes after
each successful run, and it will not repeat if $\hat B_0$ remains unchanged.
Hence the lemma holds.

We first show that for a fixed $\phi$, Algorithm 4 always
returns the same $x$ (assuming that $\hat B_0$ remains unchanged).

Recall that in Algorithm 4, we proved that if the final
flow saturates any edge $(\psi', \tau)$ that enters sink $\tau$, then $\psi'$
will newly enter $\hat B_0$. Hence if $\hat B_0$ remains unchanged, we can
ignore the capacity of those edges entering $\tau$. Also, with $\hat B_0$ fixed,
the set of edges leaving source $\sigma$ and their capacities are fixed.
Therefore, the resulting $x$ of Algorithm 4 is solely
determined by the edges between $N^+$ and $\Psi$, which is fixed by
$\phi$.\footnote{The residual graph might be different for the same $\phi$,
depending on what the min-cost flow is; but $x$ is the additional max-flow
applied to the min-cost flow, so what really determines $x$ is $\phi$.}

We then conclude that if there were two Algorithm 3 runs
that have the same resulting $\phi$, then between those two runs, there must be
no Algorithm 4 run that has positive $\Delta$ which
strictly decreases $\theta_0$.  Otherwise, we would have two
Algorithm 4 runs (after those two
Algorithm 3 runs) with the same $\phi$, where the latter
run has strictly smaller $\theta_0$ (note that $\theta_0$ never increases),
contradicting that the returning $x$ of Algorithm 4 is
completely determined by $\phi$.

Therefore, if there were two Algorithm 3 runs that result in
the same $\phi$, all Algorithm 4 runs between those two runs
must do nothing ($\Delta = 0$). Hence, $x_{\Sigma,\psi}$ is unchanged between those
two runs for all $\psi$, because only Algorithm 4 can
change $x_{\Sigma,\psi}$. 

Now consider graph $G_1$ which extends graph $G$ in the
Algorithm 2 by assuming that all edges are active
($A = N^+ \times \Psi$).  That is, for each edge leaving
source $\sigma$, its capacity is $\kappa(\sigma, i) = r_i$; for each edge
entering sink $\tau$, its capacity is $\kappa(\psi, \tau) = x_{\Sigma,\psi}$; the
weight and capacity of edge $(i, \psi)$ is $w(i,\psi) = \log(-d_{i,\psi}),
\kappa(i,\psi) = \ell(i,\psi)$ for all $i \in N^+, \psi \in \Psi$.

Then we consider the original set of active edges $A$ and make $G_2$ by
revising the following edges in $G_1$: for each edge $e = (i,\psi)$ that is not
active, set its weight $w(i,\psi) = \infty$ if $e$ is inferior, and
$w(i,\psi) = -\infty$ if $e$ is superior.

Clearly, running min-cost flow on $G_2$ would result in the same $x$ as running
Algorithm 2 because we fix non-active
edges' flow by setting their weights to $\infty$ or $-\infty$. Moreover, for the same
flow, the shortest path in the residual graph of $G_2$ is exactly the same as
the residual graph of $G$. Hence when we talk about flow or distance, they
could refer to either $G$ or $G_2$. But when we talk about the total cost
of the flow, we are referring to $G_1$, as $G_2$ has infinity cost edges and
$G$ only has active edges. That is, the cost of a flow $x$ is
$\sum_{i \in N^+, \psi \in \Psi^+} x_{i,\psi} \log(-d_{i,\psi})$

Each time that Algorithm 3 runs successfully but the
Algorithm 1 continues ($r_0$ is not used up), some
constraint of line 23 in
Algorithm 3 must be tight, which means that a new edge
$(i,\psi)$ must be entering the active edge set $A$ (otherwise we will either
continue increasing $x_{0,\psi}$ or change $\hat B_0$). 
For that newly active edge, if $x_{i,\psi} = 0$, then $dist(\psi) > dist(i) +
w(i,\psi)$ must be true (recall that $dist$ is the shortest distance from
$\psi^*$ in the residual graph) because: $\mu_{i,\psi}$'s decrease rate
$\gamma_\psi (-d_{i,\psi})$ must be slower than $\theta_i$'s decrease rate
$\gamma_i$. Similarly, if $x_{i,\psi} = \ell_{i,\psi}$, then
$dist(i) > dist(\psi) - w(i,\psi)$ must be true. Hence by adding $(i,\psi)$,
either $dist$ has to decrease or a negative cycle exists which leads to a
decrease in flow cost (w.r.t. $G_1$).

Also note that $dist$ and flow cost are completely determined by $\phi$
if $x_{\Sigma,\psi}$ is fixed for all $\psi \in \Psi$. Hence $\phi$ cannot
repeat since each $\phi$ change has either to either decrease flow cost,
or maintain the flow cost and decrease $dist$. This completes our proof.
\end{proof}

With this, we obtain an
exponential upper bound on the algorithm's runtime.  Because the algorithm
only terminates when the number of catcher resources has reached $r_0$, and
we have shown that the algorithm maintains equilibrium throughout, this
also establishes the algorithm's correctness.

\begin{theorem}
\label{theorem:bound}
Algorithm 1 computes a Nash equilibrium of the
given CE game in $2|\Psi| + 4|\Psi|3^{n|\Psi|}$ iterations.
\end{theorem}
\begin{proof}
  Because of Lemmas~\ref{lemma:no-change}, \ref{lemma:increase-ne}, and
  \ref{lemma:decrease-catcher}, Nash equilibrium is always maintained. So
  we only need to prove that the algorithm stops after at most $2|\Psi| +
  4|\Psi|3^{n|\Psi|}$ iterations. We have at most $2|\Psi|3^{n|\Psi|}$
  iterations where Algorithm 3 runs successfully, by
  Lemma~\ref{lemma:bound}.  Each failed iteration must either be followed by a
  successful iteration, or $\hat B_0$ is changed (by
  Lemma~\ref{lemma:fail-success}). Lemma~\ref{lemma:enter-once} ensures that
  $\hat B_0$ can be changed at most $2|\Psi|$ times. So overall there can be
  at most $2 \times 2|\Psi| 3^{n|\Psi|} + 2|\Psi|$ iterations.
\end{proof}

\subsection{Experiments}

Although Theorem~\ref{theorem:bound} only gives an exponential bound on the
number of iterations, the number of iterations in
Algorithm 1 grows much more slowly---about
linearly---in our experiments, as shown in
Figure~\ref{fig:exp1}. In our experiments, parameters $r$, $b$, and
$d$ are generated
uniformly at random from $\{1,\ldots,10\}$
 (or $\{-10,\ldots,-1\}$).
 Each instance of size $n$ has $n$ evaders and $n$ sites; for
 each $n$ we solved $20$ instances.
 
The running time per iteration is provably polynomial and
it grows about cubically
as Figure~\ref{fig:iter_time} shows. That is consistent with
how the network flow subroutine (which is used in each of our iterations)
typically scales.

 An alternative approach to solving for these Nash equilibria would be
 to construct the normal form of the game and use a standard NE-finding
 algorithm.  This approach, however, is doomed regardless of the
 precise choice of algorithm, because the size of the normal form blows
 up exponentially, as shown in
 %the second part of
 Figure~\ref{fig:exp2}.%\begin{fullversion}\footnote{We drew our random numbers from the
 %integers to be able to show this second figure; if they were fractions
 %it would not be clear what the normal form would be.}\end{fullversion}

Note that we implemented our algorithm completely in Python
(including the min-cost network flow subroutine).
Hence there is room
to further improve the performance by using C/C++, and/or some optimized
network flow libraries.

%Although Theorem~\ref{theorem:bound} only gives an exponential bound on the
%number of iterations, the number of iterations in
%Algorithm~\ref{algorithm:main} grows only linearly in practice, as shown in
%Figure~\ref{fig:exp}. In the experiments, parameters $r, b, d$ are generated
%uniform at random between $1$ to $10$ (or $-1$ to $-10$). The experiments only
%went up to $size=20$ because in each iteration, we need to run many
%min-cost/max flow\footnote{We used min-cost flow routine to solve max flow as
%well} routines on $2size$-node-and-$size^2$-edge graphs due to binary search.
%Note that even though the parameter value ranges only between $1$ to $10$, our
%binary search is set to have an accuracy of $10^{-10}$ because players can
%allocate fractional resources.
%

\section{Future Research}

The obvious direction for future research is resolving whether our
algorithm in fact provably runs in polynomial time---and, if not,
whether there is another algorithm that does.  The algorithm's success
in experiments gives us hope that the answer to at least one of these
two questions is positive, but we have not been able to decisively
answer them.  There are several indications that the question is
inherently difficult to answer.  The earlier algorithm for multiple
attacker resources in the non-Bayesian case and the proof of the
polynomial bound on its runtime~\cite{Korzhyk11:Security} were
already quite involved, and we showed that the Bayesian case requires
us to deal with additional challenging issues
(Subsection ``Reducing from Min-Cost Matching'').  Still, we believe that the importance
of solving Bayesian security games would justify the devotion of
further effort to resolving this question, as well as to extending
these techniques to related problems.

%Acknowledgments are optional
\section*{Acknowledgments}
We thank Ronald Parr for his contributions to our early discussions about
Bayesian security games.
We are also thankful for support from ARO under grants W911NF-12-1-0550 and
W911NF-11-1-0332, NSF under awards IIS-0953756, IIS-1527434, CCF-1101659, and CCF-1337215, and a Guggenheim Fellowship. 
Part of this research was done while Conitzer was visiting the Simons Institute for the Theory of Computing.

%\balancecolumns

%
% For AAMAS-2016, as references are unlimited but appendices must fit within
% 8 pages, the References section must come after the appendices (if any)
%\bibliographystyle{abbrv}
%\bibliography{references}  % sigproc.bib is the name of the Bibliography in this case
% ACM needs 'a single self-contained file'!
%\balancecolumns % GM June 2007
% That's all folks!

\end{document}